\documentclass[11pt]{article}

\usepackage{amsmath, amsfonts, amssymb, amsthm}
\usepackage{titlesec}
\titleformat{\paragraph}[runin]{\normalfont\itshape}{\theparagraph.}{.3em}{}[.]\titlespacing{\paragraph}{0pt}{1ex plus .1ex minus .2ex}{.5em}

\usepackage{enumerate}
\usepackage[usenames,dvipsnames]{xcolor}

\usepackage[colorlinks=true, linkcolor=Blue, citecolor=Blue,pdfborder={0 0 0}]{hyperref}
\usepackage{url}
\usepackage{mathtools}
\usepackage{ifsym}
\usepackage{multirow}
 
\usepackage{cite}
\usepackage{cleveref}
\usepackage[normalem]{ulem} 
\usepackage{cancel} 

\usepackage[margin=3cm]{geometry}

\newif\iffinal
\finalfalse 
\finaltrue 
\iffinal\else\usepackage[notref,notcite]{showkeys}\fi

\newtheorem{theorem}{Theorem}[section]
\newtheorem{lemma}[theorem]{Lemma}

\newtheorem{proposition}[theorem]{Proposition}
\newtheorem{definition}[theorem]{Definition}

\makeatletter
\def\namedlabel#1#2{\begingroup
	\def\@currentlabel{#2}%
	\label{#1}\endgroup
}
\makeatother
\newenvironment{assumption}[2]{\par\vspace{1em}\noindent\textbf{Assumption #1.}\namedlabel{#2}{\textbf{#1}}}{\par\vspace{1em}}

\crefname{theorem}{Theorem}{Theorems}
\crefname{lemma}{Lemma}{Lemmas}
\crefname{proposition}{Proposition}{Propositions}
\crefname{section}{Section}{Sections}
\crefname{assumption}{Assumption}{Assumptions}
\crefname{equation}{Eq.}{Eqs.}

\newcommand{\Tr}{\operatorname{Tr}}

\newcommand{\ket}[1]{|{#1}\rangle}
\newcommand{\bra}[1]{\langle{#1}|}

\renewcommand{\Re}{\operatorname{Re}}
\renewcommand{\Im}{\operatorname{Im}}

\newcommand\otimesal{\mathop{\hbox{\raise 1.6 ex
  \hbox{$\scriptscriptstyle\mathrm{al}$}
\kern -0.92 em \hbox{$\otimes$}}}}
\newcommand\oplusal{\mathop{\hbox{\raise 1.6 ex
  \hbox{$\scriptscriptstyle\mathrm{al}$}
\kern -0.92 em \hbox{$\oplus$}}}}
\newcommand\Gammal{\hbox{\raise 1.7 ex
\hbox{$\scriptscriptstyle\mathrm{al}$}\kern -0.50 em $\Gamma$}}
\renewcommand\i{\mathrm{i}}

\newcommand{\e}{{\mathrm e}}
\newcommand{\iu}{{\mathrm i}}
\renewcommand{\d}{{\mathrm d}}

\newcommand{\tr }{\mathrm{tr}}

\newcommand{\id}{\mathrm{id}}

\author{Miguel Ballesteros$^1$, Tristan Benoist$^2$, Martin Fraas$^{3,4}$, J\"urg Fr\"ohlich$^5$}

\title{The appearance of particle tracks in detectors}

\begin{document}

\maketitle

\vspace{1em}

\begin{abstract}
The phenomenon that a quantum particle propagating in a detector, such as a 
Wilson cloud chamber, leaves a track close to a classical trajectory is analyzed. We introduce an idealized quantum-mechanical model of a charged particle that is periodically illuminated by pulses of laser light resulting in repeated indirect
measurements of the approximate position of the particle. For this model we present a mathematically rigorous analysis 
of the appearance of particle tracks, assuming that the Hamiltonian of the particle is quadratic in the position- and 
momentum operators, as for a freely moving particle or a harmonic oscillator.
\end{abstract}

\section{Introduction: The problem, its history, and a heuristic discussion}\label{Intro}
In this paper we present a quantum-mechanical analysis of the propagation of a charged particle in a detector tracking its approximate position. As first proposed by \textit{Gamow}, the decay of a radioactive nucleus resulting in the emission of an 
$\alpha$-particle can be understood by invoking the \textit{tunnel effect}; see \cite{Gamow}. If the radioactive nucleus 
is very heavy one may assume that its center-of-mass position is fixed (e.g., at the origin). The initial state of the 
$\alpha$-particle is then given by an s-wave, i.e., it is perfectly spherically symmetric. However, apparently 
thanks to interactions between the $\alpha$-particle and degrees of freedom of the detector, the $\alpha$-particle 
leaves the nucleus in a fairly well-defined direction and then propagates along a trajectory 
that is close to a solution of a \textit{classical} equation of motion for a charged point-particle. 
Thus, the spherical symmetry of the initial state is \textit{broken}, and the behavior of the $\alpha$-particle in the detector is ``particle-like'' rather than ``wave-like''. The situation would be entirely different if the detector were absent and the particle propagated in empty space: The behavior of the $\alpha$-particle would then be ``wave-like'', which could be verified, in principle, by diffraction of the 
particle-wave in an array of gratings placed quite far from the nucleus.

 The appearance of particle tracks in detectors, such as a Wilson cloud chamber, which signals the breaking of 
 a symmetry of the initial state and of the dynamics of the charged particles, has puzzled the founding fathers of 
 quantum mechanics since the late twenties of the past century. It may be of interest to sketch some historical facts 
 concerning studies of this phenomenon, for which we rely on a very informative 2013 paper by R. Figari and A. Teta \cite{FiTe}.  
 
 In 1927, during a famous Solvay conference, \textit{Einstein} apparently drew attention to the puzzle of how the appearance of 
 particle tracks in a cloud chamber can be understood quantum-mechanically. \textit{Born} summarized Einstein's question as follows: \textit{``A radioactive sample emits $\alpha$-particles in all directions; 
 these are made visible by the method of the Wilson cloud chamber. Now, if one associates a spherical wave with each 
 emission process, how can one understand that the track of each $\alpha$-particle appears as a (very nearly) straight line? 
 In other words: how can the corpuscular character of the phenomenon be reconciled here with the representation by waves?''} 
 Born suggested that the answer to this question can be found in the phenomenon of \textit{``reduction of the probability packet''}, 
 as discussed by \textit{Heisenberg} in 1927. Heisenberg considers an experiment where the approximate 
 position of a charged particle
  is determined by scattering light of wavelength $\lambda$ off the particle, resulting in a collapse of its 
  position-space wave function to one whose support is essentially contained in a region of diameter 
  $\sim \mathcal{O}(\lambda)$ indicating the instantaneous position of the particle. In Heisenberg's words, 
  \textit{``every determination of position reduces therefore the wave packet back to its original size $\lambda$''}. 
  This mechanism could in principle explain the appearance of nearly classical trajectories traced out by a charged particle 
  in an experiment where it is irradiated periodically by a coherent flash of light of a wavelength  
  $\lambda$ sufficiently big so as not to change the velocity of the particle significantly; 
  (i.e., $\frac{h}{M\cdot \lambda} \ll v_0$, where $h$ is Planck's constant, $M$ is the mass of the particle, and $v_0$ 
  is its initial speed). The scattered light could then be observed in detectors, and this would 
  provide information on the approximate position of the charged particle.
  
Subsequently, a debate emerged as to whether one should describe the measurements of the approximate positions of the 
charged particle abstractly, as \textit{direct measurements}, invoking the ``Copenhagen interpretation'' of 
quantum mechanics including the ``collapse of the wave function'' in a measurement, or by working out the 
quantum mechanics of the particle interacting with the degrees of freedom 
of the detector -- in the example treated by Heisenberg the photons constituting the flashes of light, which are finally 
observed directly in detectors (e.g., in photomultipliers). No matter which point of view one adopts in coping with
this question, at some stage one has to invoke (a theory of) \textit{direct measurements}
\footnote{The \textit{``ETH Approach'' to quantum mechanics} \cite{Fr} provides such a theory.} 
in order to complete the argument. (In the Copenhagen interpretation of quantum mechanics, the stage, 
in the analysis of an experiment, where direct measurements and the collpase of the wave function are invoked 
depends on where one places the so-called Heisenberg cut.) Moreover, for all practical purposes, the predictions 
about the propagation of the charged particle one comes up with ought to be independent of which point 
of view one adopts. In other words, it should not depend on where exactly one places the Heisenberg cut; as already discussed by von Neumann in his celebrated book \cite{vN}. 

A famous analysis of the tracks created by an $\alpha$-particle propagating in a Wilson chamber has been 
carried out by \textit{ Mott} in 1929 \cite{Mott}, inspired by insightful suggestions by \textit{Darwin} \cite{Darwin}. 
Mott takes into account the effect of interactions between the $\alpha$-particle and the gas atoms in the 
Wilson chamber. Using second-order stationary perturbation theory, he estimates the probability that two gas 
atoms are excited/ionized by the passage of an $\alpha$-particle and shows that this probability is substantially different 
from zero if and only if the two gas atoms are located in a cone of small opening angle emanating from the position 
of the radioactive nucleus that emits the $\alpha$-particle.\footnote{Darwin's paper may actually be more interesting 
than Mott's. For, Darwin describes ideas that amount to an outline of what is now called ``decoherence''.} For a more precise review and extension of these ideas we refer the reader to an interesting recent booklet by Figari and Teta \cite{FigTeta}.

\textit{Remark:} The appearance of particle tracks in detectors that monitor the approximate position of charged particles
exemplifies the following general phenomena: Consider the quantum mechanics of a physical system that exhibits a
 dynamical symmetry described by some group $G$. We imagine that the system is prepared in an initial state that is
 \textit{invariant} under the action of $G$ on its state space and that, subsequently, an observable, $O$, is measured that
transforms \textit{non-trivially} under $G$. Then, after a successful direct or indirect measurement of $O$, the state of the
system is \textit{no longer invariant} under the action of $G$. In the example of the particle tracks, $G$ is the group of
space-rotations, and $O$ is the position operator of the charged particle. In easier examples, $G$ is a finite group and
$O$ is an observable with discrete spectrum.\footnote{See \cite{MK, BB, BBB, BFFS} for a theory of indirect measurements in quantum mechanics.} The fact that, in quantum mechanics, symmetries can be broken 
(and the associated conservation laws are \textit{violated}) in measurements apparently tends to
puzzle people. -- Well, it should not! It is a common phenomenon and can be understood on the basis of a 
good quantum theory of measurements. 

The second general issue illustrated by the analysis presented in this paper concerns
 the question of how approximate values of \textit{non-commuting} observables, in this paper the position and velocity 
 of a particle, can be inferred from measurements of appropriate quantities in such a way that the uncertainties of their 
 values are close to the minimal uncertainties compatible with \textit{Heisenberg's uncertainty relations}. 
 
In this paper, we present a theoretical analysis of a gedanken experiment of the sort Heisenberg  had in mind in 1927: 
We imagine that a charged particle, an $\alpha$-particle or an electron, is prepared in a fairly arbitrary initial state, $\Psi$, 
whose support in x-space is localized close to the origin of the coordinate system to be used. The particle's position- and 
momentum operators are denoted by ${\bf{X}}$ and ${\bf{P}}$, respectively. Let $M$ denote the mass of the particle, 
and let 
\begin{equation}\label{initial state}
v_0:= \frac{1}{M} \sqrt{\langle \Psi, {\bf{P}}^{2} \Psi \rangle}\,,
\end{equation}
be a measure for the initial speed of the particle, with $\langle \cdot ,\cdot \rangle$ the scalar product on the Hilbert space of pure state vectors of the particle. We imagine that, 
every $\tau$ seconds, the particle is irradiated by a pulse of coherent light of wave-length $\lambda$, with 
\begin{equation}\label{speed}
\frac{h}{M\cdot \lambda} \ll v_0\,.
\end{equation}

The scattered light then triggers the firing of an array of photomultipliers, which results in an approximate \textit{indirect} measurement of the position of the particle; (see, e.g., \cite{BCFFS} for an analysis of indirect measurements of observables with continuous spectra, such as ${\bf{X}}$). The indirect measurement of the 
particle's position has a precision of $\mathcal{O}(\lambda)$, and the change of the velocity of the particle in such a 
measurement is of $\mathcal{O} \big(h/(M \cdot \lambda)\big)$, i.e., very small, for a very large particle mass $M$, 
see \eqref{speed}. (Condition \eqref{speed} is met in observations of $\alpha$-decay of radioactive nuclei using scattering of visible light.) Suppose the particle is found near a point ${\bf{x}}_1$ at time $0$ and near a point ${\bf{x}}_2$ 
at time $\tau$. Then it is likely that, at some later time $t$, the particle will be found near a point 
$${\bf{x}}(t) \sim {\bf{x}}_1+ t\, {\bf{v}},$$ 
where ${\bf{v}}\sim \tau^{-1} [{\bf{x}}_2 - {\bf{x}}_1].$ This is, roughly speaking, the assertion that we want to establish (in a simple model), in this paper. Our results thus furnish an example of predicting the approximate location of a 
quantum particle in classical phase space at different times when the particle's approximate position 
is measured periodically. The results established in this paper for a certain class of models that can essentially be solved exactly hold for arbitrary particle masses; i.e., it is not necessary to approach a semi-classical regime for our results to hold. But the smaller the particle mass is the larger the scatter of observed positions around a classical particle trajectory will turn out to be.
 
\subsection{The model to be studied}\label{mod}
Next, we describe the main features of a model of the quantum theory of particle tracks that we propose to analyze in this paper; (for more details see Section 2). The Hilbert space of pure state vectors of the charged particle is chosen to be
\begin{equation}\label{Hilbert}
\mathcal{H}:= L^{2}(\mathbb{R}^{d}, \d^{d}x)\,,
\end{equation}
where $d$ is the dimension of physical space, $\mathbb{R}^{d}$, (with $d=2$ or $3$, in a realistic model).
General states of the particle are given by \textit{density matrices}, i.e., by non-negative, trace-class 
operators of trace one acting on $\mathcal{H}$. The Hilbert space of state vectors of the subsystem consisting of the electromagnetic field and the photomultipliers is given by a direct integral of fibre spaces 
$$\mathfrak{H} = \int^{\oplus}_{\mathbb{R}^{d}} d^{d}{\bf{q}}\, \mathfrak{H}_{{\bf{q}}}\,, \qquad \text{with  }\,\, 
\mathfrak{H}_{{\bf{q}}} \simeq \mathbb{C}^{N}, \, \, N\leq \infty\,,$$ 
where ${\bf{q}}$ is a point in the spectrum, $\mathbb{R}^{d}$, of a vector, ${\bf{Q}}$, of commuting selfadjoint operators $Q_{j}, j=1,\dots,d,$ whose measured values, ${\bf{q}} \in \mathbb{R}^{d}$, are supposed to be tightly correlated 
with the positions, ${\bf{x}} \in \mathbb{R}^{d}$, of the charged particle. In the following, we assume that the spectrum of 
${\bf{Q}}$ has multiplicity 1, so that $\text{dim } \mathfrak{H}_{{\bf{q}}} = 1, \, \forall {\bf{q}} \in \text{spec}({\bf{Q}})\,, (i.e., 
N=1)$. But this assumption is not essential; it merely serves to keep our calculations as simple and transparent as possible.

We consider a state 
$\Psi \otimes \Omega \in \mathcal{H}\otimes \mathfrak{H}$. The propagator $U_{t}({\bf{X}})$ describing the time evolution of the electromagnetic field and the photomultipliers during the light-scattering process is defined by
\begin{equation}\label{prop}
\Big[U_{t}({\bf{X}}) \big( \Psi \otimes \Omega \big)\Big]({\bf{x}}):= \Psi({\bf{x}})\cdot U_{t}({\bf{x}})\Omega\,, \qquad \text{for  }\,\, {\bf{x}}\in \mathbb{R}^{d}= \text{spec}({\bf{X}})\,,
\end{equation}
where $\mathbf{X}$ is the position operator of the particle, as above, and the operators $U_{t}({\bf{x}})$ 
are unitary on $\mathfrak{H}$, for almost all ${\bf{x}}\in \text{spec}(\mathbf{X}) \simeq \mathbb{R}^{d}$, 
with ${\bf{x}}$ the position of the charged particle during the scattering process. 
The parameter $t\ll \tau$ is the time of propagation elapsing between the moment when the photons are scattered 
off the charged particle and the moment when they cause the photomultipliers to fire. We suppose that $t$ is so small 
that the position of the charged particle is approximately constant in a time interval of length $t$; i.e., one supposes that 
$t\cdot v_0 \ll \lambda$. We assume that the operators $U_{t}({\bf{x}})$ depend ``sensitively'' on ${\bf{x}} \in \mathbb{R}^{d}$ (see Eqs. \eqref{trans ampl} and \eqref{Gauss}, below).

In our model the state of the electromagnetic field and of the photomultipliers just before a light pulse is generated is 
always given by the \textit{same} unit vector $\Omega_{in} \in \mathfrak{H}$, and that, after a scattering process triggering the 
firing of the photomultipliers, the state of the electromagnetic field and of the photomultipliers relaxes back to $\Omega_{in}$, 
with a relaxation time that is small as compared to the time $\tau$ between two subsequent light-scattering processes. 
In every scattering process, the scattered light triggers an array of photomultipliers to fire, and this is assumed to amount 
to a measurement of the approximate value of the operator ${\bf{Q}}$. 
According to the standard postulate on the collapse of the wave function, this measurement results in a final state of the form, 
 \begin{equation}\label{final state}
 \Omega_{fin} = \int_{\Delta}^{\oplus} \d^{d}{q} \, h(\mathbf{q})\, \Omega({\bf{q}})\,,  \qquad \int_{\Delta} \d^{d}q\,\vert h(\mathbf{q})\vert^{2} =1\,,
 \end{equation}
with $ \Vert \Omega({\bf{q}}) \Vert_{\mathfrak{H}_{{\bf{q}}}} = 1\,, \,\,\forall {\bf{q}}$, assuming that the measured 
value of ${\bf{Q}}$ is contained in a cell $\Delta \subset \mathbb{R}^{d}$ (of diameter 
$\mathcal{O}(\lambda)$). Since ${\bf{Q}}$ is a vector of commuting self-adjoint operators, the final states, i.e., 
the generalized eigenstates of ${\bf{Q}}$, are complete, in the sense that
\begin{equation}\label{completeness-1}
\int_{\mathbb{R}^{d}}^{\oplus} \d^{d}q\, \vert \Omega({\bf{q}})\rangle \langle\Omega({\bf{q}})\vert  = {\bf{1}}\,.
\end{equation}
We define the transition amplitudes
\begin{equation}\label{trans ampl}
 V_{\bf{q}} (\mathbf{x}):= \frac{\langle \Omega({\bf{q}})\,\d^{d}{q}, U_{t}({\bf{x}}) \Omega_{in}\rangle}{\d^{d}{q}} \,.
 \end{equation}
The idea is now that the cell $\Delta$ in Eq. \eqref{final state} corresponds to the approximate position of the charged particle, as inferred \textit{indirectly} from the firing pattern of the photomultipliers. Since the wave length of the light pulse scattering off the charged particle is given by 
$\lambda$, we expect that $\vert V_{\bf{q}} ({\bf{x}}) \vert$ is very small, unless $\vert {\bf{q}}-{\bf{x}} \vert \leq \lambda$. Since the momentum transfer from the photons to the charged particle is neglibily small, as compared to $Mv_0$, the \textit{phase} 
of the transition amplitude $V_{\bf{q}}({\bf{x}})$ is approximately ${\bf{x}}$-\textit{independent}. Completeness of the final states of the electromagnetic field and the photomultipliers, i.e., Eq. \eqref{completeness-1}, implies that
\begin{equation}\label{completeness-2}
\int d^{d}q\, V_{\bf{q}} ({\bf{x}})^{*}\cdot V_{\bf{q}} ({\bf{x}}) = \Vert U_{t}({\bf{x}}) \Omega_{in} \Vert^{2} 
= \Vert \Omega_{in} \Vert^{2}\overset{!}{=} 1\,,
\end{equation}
where we have used that the propagator $U_{t}({\bf{x}})$ is a unitary operator. In this paper, we consider the following somewhat simplistic ansatz 
\begin{equation}\label{Gauss}
V_{{\bf{q}}} ({\bf{X}})= (\sqrt{2\pi}^{d}\cdot \lambda)^{-\frac{1}{2}}\, \cdot\text{exp}\Big[-\frac{\vert {\bf{X}}-{\bf{q}} \vert^{2}}{4\lambda^{2}}\Big]\,,
\end{equation}
which has all the desired properties described above.
As already mentioned, we assume that the duration, $t$, of the light scattering process, including the firing of the 
photomultipliers, is small as compared to the waiting time, $\tau$, between two light pulses, i.e., $t\ll \tau$, and that 
$v_0 \cdot t \ll \lambda$. Then the motion of the charged particle during a single light scattering process can be neglected.
In between two such processes the state vector of the particle propagates according to the Schr\"odinger equation 
$$\Psi \mapsto U_{S} \Psi\,,$$
where $U_{S}$ is the propagator of the particle corresponding to a time difference $\tau$. If the initial state of the particle is a mixture described by a density matrix $\rho$ acting on $\mathcal{H}$ then the state after a time $\tau$ is given by 
$$\rho_{\tau}:= U_{S}\, \rho\, U_{S}^{*}\,.$$
If the Hamiltonian of the particle is quadratic in its position operator ${\bf{X}}$ and its momentum operator ${\bf{P}}$, 
which will be assumed in the following, then one has that, in the Heisenberg picture,
\begin{equation}\label{evol}
U_{S}^{*} \begin{pmatrix}{\bf{X}}\\{\bf{P}}\end{pmatrix} U_{S}=S \begin{pmatrix} {\bf{X}}\\ {\bf{P}} \end{pmatrix}\,,
\end{equation}
where $S: \Gamma \rightarrow \Gamma$ is a symplectic matrix on the phase space, $\Gamma$, of the particle, with 
$\Gamma:=\mathbb{R}^{d}_{x} \oplus \mathbb{R}^{d}_{p}$. 
The meaning of the left side of Eq. \eqref{evol} is that 
$$U_{S}^{*} \begin{pmatrix}{\bf{X}}\\{\bf{P}}\end{pmatrix} U_{S} := 
\begin{pmatrix}U_{S}^{*}X_1 U_{S},&\dots&,U_{S}^{*}X_d U_{S}\,,\,U_{S}^{*}P_1 U_{S},&\dots&,U_{S}^{*}P_{d} U_{S}              \end{pmatrix}^{t}\,,$$  
where $(\cdot)^{t}$ denotes transposition. For a freely moving particle,
\begin{equation}
S\equiv \begin{pmatrix} S_{xx}&S_{xp}\\ S_{px}& S_{pp} \end{pmatrix} :=\begin{pmatrix} {\bf{1}} & \frac{\tau}{M} {\bf{1}}\\
0&\,\, \,{\bf{1}}\end{pmatrix}\,,
\end{equation}
where ${\bf{1}}$ is the identity matrix on $\mathbb{R}^{d}$; hence the Heisenberg evolution of the position- and momentum operator of a freely moving particle in a time step of length $n\tau, n=0, 1, 2, \dots,$ is given by
\begin{equation}\label{Heisenberg}
\begin{pmatrix}{\bf{X}}\\{\bf{P}}\end{pmatrix} \mapsto \begin{pmatrix}{\bf{X}}_{n\tau}\\ {\bf{P}}_{n\tau} \end{pmatrix} := 
\begin{pmatrix}{\bf{X}}+ n \tau {\bf{V}}\\{\bf{P}}\end{pmatrix}\,, \quad \text{with  }\,\, {\bf{V}}:= \frac{1}{M}{\bf{P}}\,.
\end{equation}
The vector operator ${\bf{V}}$ is the velocity operator. We observe that, for a freely moving particle, the operators
\begin{equation}\label{conserved}
{\bf{V}}_{n\tau} \equiv {\bf{V}} \quad \text{ and } \quad {\bf{X}}_{n\tau} - n\tau {\bf{V}}\equiv {\bf{X}}
\end{equation}
are \textit{conservation laws}.
In terms of the components $X_i$ and $V_i$, $i=1,\dots, d,$ of the operators ${\bf{X}}$ and ${\bf{V}}$, respectively, Heisenberg's commutation relations read
\begin{equation}\label{CCR}
\big[X_i, V_j\big]=  \frac{\hbar}{M} \delta_{ij}\,, \,\, \,\big[X_i, X_j\big]=\big[V_i, V_j\big]=0\,, \quad \forall\,\, i, j= 1,\dots, d.
\end{equation}

Let us now suppose that, right after preparing the particle in a specific state, its position is measured indirectly and approximately by scattering a pulse of coherent light off the particle, as described above, with the result that the measured $d$-tuple, ${\bf{q}}$, of values of the vector operator ${\bf{Q}}$ is found to belong to a subset $\Delta_0 \subset \mathbb{R}^{d}$. By Eq. \eqref{Gauss} this means that the particle is found within a distance of $\mathcal{O}(\lambda)$ of the set $\Delta_0$. Subsequently, the particle propagates freely for a time $\tau$, whereupon its position is measured indirectly again, with ${\bf{q}} \in \Delta_1$, for some subset $\Delta_1$ of $\mathbb{R}^{d}$, etc.. If the initial state of the charged particle is given by a density matrix $\rho$ then its state after the first indirect position measurement is given by
\begin{equation}\label{collapse}
\rho \mapsto \frac{\int_{\Delta_0} \d^{d}{q}\, V_{{\bf{q}}}({\bf{X}}) \, \rho \, V_{{\bf{q}}}({\bf{X}})^{*}}{\int_{\Delta_0} 
\d^{d}{q}\, \text{tr}\Big[ V_{{\bf{q}}}({\bf{X}}) \, \rho \, V_{{\bf{q}}}({\bf{X}})^{*}\Big]} \,,
\end{equation}
where we have used the ``collapse postulate''.
Subsequently, the particle propagates freely for a time $\tau$, whereupon its position is measured indirectly again, etc.; i.e., 
we suppose that the position of the particle is measured repeatedly and found to correspond to measurement outcomes 
${\bf{q}}_j\in \Delta_j$ at times $j\tau$, $j=0, 1, 2, \dots, n$. We define an evolution operator $W_{n}(\underline{{\bf{q}}}_{n})$ by
\begin{equation}\label{history}
W_{n}(\underline{{\bf{q}}}_{n}):= U_{S^{n+1}}V_{{\bf{q}}_{n}}({\bf{X}}_{n\tau}) \dots V_{{\bf{q}}_{0}}({\bf{X}})\,.
\end{equation}
After $n+1$ indirect position measurements, as described above, the state of the particle is given by
\begin{equation}\label{measured state}
\rho_{\underline{\Delta}_n}:= \frac{\int_{\Delta_0} \d^{d}{q}_0 \dots \int_{\Delta_n} \d^{d}{q}_n \, W_{n}(\underline{{\bf{q}}}_{n})\, \rho\, W_{n}(\underline{{\bf{q}}}_{n})^{*}}{\int_{\Delta_0} \d^{d}{q}_0 \dots \int_{\Delta_n} 
\d^{d}{q}_n \, \text{tr}\Big[ 
W_{n}(\underline{{\bf{q}}}_{n})\, \rho\, W_{n}(\underline{{\bf{q}}}_{n})^{*} \Big]}
\end{equation}
 We introduce a Borel measure 
 \begin{equation}\label{prob measure}
 \d \mathbb{P}^{(n)}_{\rho}(\underline{{\bf{q}}}_{n}):= 
 \text{tr}\Big[ W_{n}(\underline{{\bf{q}}}_{n})\, \rho\, W_{n}(\underline{{\bf{q}}}_{n})^{*} \Big] \prod_{j=0}^{n} \d^{d}{q}_j = 
 \text{tr}\Big[  \rho\,\, W_{n}(\underline{{\bf{q}}}_{n})^{*} \cdot W_{n}(\underline{{\bf{q}}}_{n}) \Big] \prod_{j=0}^{n}
 \d^{d}{q}_j \,.
 \end{equation}
 defined on $\big(\mathbb{R}^{d}\big)^{\times (n+1)}$. It follows from repeated application of 
 Eq. \eqref{completeness-2} that 
$$ \int_{(\mathbb{R}^{d})^{\times (n+1)}}  \d \mathbb{P}^{(n)}_{\rho}( \underline{{\bf{q}}}_{n}) =1\,,$$
i.e., $ \d\mathbb{P}^{(n)}_{\rho}(\underline{{\bf{q}}}_{n})\equiv \mathbb{P}^{(n)}_{\rho}(\d \underline{\mathbf{q}}_n)$ 
is a \textit{probability measure}. It has the property that
 $$\int_{\mathbb{R}^{d}_{{\bf{q}}_n}} \mathbb{P}^{(n)}_{\rho}(\d \underline{{\bf{q}}}_{n-1}, \d{\bf{q}}_{n}) = 
\mathbb{P}^{(n-1)}_{\rho}(\d \underline{{\bf{q}}}_{n-1})\,,$$
 which, according to a lemma due to \textit{Kolmogorov}, implies that there exists a measure 
 $\d \mathbb{P}_{\rho}(\underline{{\bf{q}}}_{\infty})$ on the space, 
 $\mathfrak{Q}:= \bigtimes_{j=0}^{\infty} \mathbb{R}_{j}^{d}$, of infinite sequences 
 $\underline{{\bf{q}}}_{\infty} \in \mathfrak{Q}$ whose marginal in the first $n+1$ arguments is given by 
 $d\mathbb{P}^{(n)}_{\rho}(\underline{{\bf{q}}}_{n})$. 
 We conclude that
 \begin{equation}\label{probas}
 \text{Prob}_{\rho}(\Delta_0, \dots, \Delta_n):= \int_{\Delta_0} \dots \int_{\Delta_n} \d\mathbb{P}^{(n)}_{\rho}(\underline{{\bf{q}}}_{n})
 \end{equation}
 is the probability of the event that the position of the charged particle at time $j\tau$ is within a distance of $\mathcal{O}(\lambda)$ of the subset $\Delta_j \subset \mathbb{R}^{d}$, \,$\forall \,j=0,1,\dots, n$. Indeed, it follows from Eqs. \eqref{completeness-2} and \eqref{prob measure} that
 \begin{equation}
 0 \leq \text{Prob}_{\rho}(\Delta_0, \dots, \Delta_n) \leq 1, \quad \text{  with  }\quad 
 \text{Prob}_{\rho}(\mathbb{R}^{d}, \dots, \mathbb{R}^{d}) =1\,.
 \end{equation}

 The idea is now to verify that, with high probability, the cells $\Delta_0,\dots, \Delta_n$, which indicate the positions of the particle at times $0, \tau, \dots, n\tau$, are centered in points ``close'' to ${\bf{x}}(0), {\bf{x}}(\tau), \dots , {\bf{x}}(n\tau)$, respectively, where 
 ${\bf{x}}(t)={\bf{x}}+t {\bf{v}},\, t\in [0, n\tau],$ is the trajectory of a freely moving classical particle. Before entering into the intricacies of rigorous proofs, we propose to explain heuristically why this must be correct.
 
 \subsection{Understanding the origin of particle tracks in a model of an optical detector}
 It is plausible to expect that the reason why particle tracks are observed when the approximate particle position is measured repeatedly lies in certain properties of the measure $\d\mathbb{P}_{\rho}(\underline{{\bf{q}}}_{\infty})$. We have seen that, for large values of the particle mass $M$, the components of the particle position operator ${\bf{X}}$ almost commute with all components of the particle velocity operator ${\bf{V}}$; see \eqref{CCR}. This implies that, in the limit where $M \rightarrow \infty$, the operators 
 $V_{{\bf{q}}_{j}}({\bf{X}}_{j \tau})\,, j= 0,1,\dots, n,$ all commute, i.e.,
 \begin{equation}\label{class lim}
 \Big[V_{{\bf{q}}_{i}}({\bf{X}}_{i \tau}), V_{{\bf{q}}_{j}}({\bf{X}}_{j \tau}) \Big] \rightarrow 0\,, \quad \text{as   }\,\, M\rightarrow \infty\,, \,\, \forall \,\, i,j=1,\dots, n\,.
 \end{equation}
 Thus,
 \begin{equation}\label{lim measure}
 \d\mathbb{P}^{(n)}_{\rho}(\underline{{\bf{q}}}_{n}) \rightarrow \text{tr}\Big( \rho\,\prod_{j=0}^{n} \vert V_{{\bf{q}}_j}({\bf{X}}+j\tau {\bf{V}}) \vert^{2}\Big) \d^{d} q_{j}, \quad \text{ as  }\,\, M\rightarrow \infty,
 \end{equation}
for arbitrary $n=0,1,2,\dots$. We conclude that, as $M\rightarrow \infty$, the measure 
$\d \mathbb{P}_{\rho}(\underline{{\bf{q}}}_{\infty})$ approaches a convex combination of
 \textit{product measures} on the space $\mathfrak{Q}$ of infinite sequences, which
depend on the conservation laws ${\bf{X}}$ and ${\bf{V}}$.
If the transition amplitudes $V_{{\bf{q}}}({\bf{X}})$ are chosen as in Eq.~\eqref{Gauss} then, in the limit of infinitely many indirect measurements of the approximate particle position, i.e., when $n \rightarrow \infty$, 
and for a very large particle mass, the 
state of the particle ``purifies'' on \textit{increasingly precise values} of the conservation laws ${\bf{X}}$ and $
{\bf{V}}$; i.e., the motion of the particle approaches a \textit{classical trajectory}
${\bf{x}}_{n\tau} = {\bf{x}} + n\tau {\bf{v}}\,,$ where ${\bf{x}}$ and ${\bf{v}}$ are points in the spectrum of the vector 
operators ${\bf{X}}$ and ${\bf{V}}$, which commute in the limit of an infinitely heavy particle; see Eqs. \eqref{conserved}, \eqref{CCR}. In the limit where $M \rightarrow \infty$, this actually turns out to be a theorem; see \cite{BBB, BCFFS}. 
 
 For finite values of the mass $M$, the operators  $V_{{\bf{q}}_{j}}({\bf{X}}_{j \tau})\,, j= 0,1,\dots, n,$ do not commute with one another; but the norms of the commutators are small, for large values of $M$. Thus, the measures 
 $\d\mathbb{P}^{(n)}_{\rho}(\underline{{\bf{q}}}_{n})$ aren't quite convex combinations of 
 product measures; but they are close to the right side of \eqref{lim measure}, provided $M$ is large. 
 The states occupied by the particle are then expected to approach a state with small
position- and momentum uncertainties, $\Delta {\bf{X}} \sim \lambda$ and $\Delta {\bf{P}} \sim \frac{h}{\lambda}$, 
as the number, $n$, of indirect measurements of the approximate particle positions tends to infinity. 
For a class of models slightly generalizing the one discussed above (see \eqref{Gauss}), 
a precise result of this sort is established in the remaining sections of this paper. Our proof relies on the following
\begin{center} 
\textit{Key identities:} Eqs. \eqref{eq:rec_fwd_zeta} through \eqref{dec-of-1}, Sect. \ref{sec:SMR}, with derivations 
presented in\\ Sect. \ref{sec:ECS} and estimates given in Sect. \ref{sec:proof_estim}. 
\end{center}
That certain stochastic Schr\"{o}dinger equations with quadratic Hamiltonians, but different from the models we consider, have solutions in closed form if the initial state is coherent has previously been observed in \cite{BK, BDK, BBJ}. In \cite{BBJ}, the authors conjecture, more generally, how the solutions behave in a semiclassical regime for Hamiltonians with more general (non-quadratic) potentials. Although the realization that our evolution equations can be solved exactly for coherent states  plays an important role in our analysis, our focus is actually on the emergence of particle tracks for \textit{arbitrary initial states}.

We expect that results of the type described above hold for a considerably more general class of models. 
(A fact potentially useful for a generalization of our analysis is a theorem due to \textit{Jean Bourgain} \cite{Bourgain} 
that says that $L^{2}(\mathbb{R}^{d}, \d^{d}x)$ admits an orthonormal basis consisting of functions with uniformly 
bounded position- and momentum uncertainties.  A tentative application of this result to the problem of particle 
tracks has been outlined in \cite{BFF}.)

\bigskip
\noindent
{\bf{Organization of the paper}.}\\
Our paper is organized as follows. In Sect.~\ref{models}, we describe the models studied in this paper in a precise way, we describe the time evolution of the particle subjected to repeated measurements of its approximate position (see \eqref{history}),
and we summarize our main results; see Theorems \ref{thm:Q_and_AR}, \ref{estimation}, \ref{thm:AR}, and \ref{thm:POVM_equivalence}.
In Sect.~\ref{sec:SMR}, we state two key lemmas this paper is based on. They lead directly to the results 
summarized in Sect.~\ref{models}. We outline the proof of one of these lemmas (the proof of the other lemma being deferred to an appendix). Technical ingredients and tools needed in our analysis appear in the remaining sections.\\
Our results are based on the use of coherent states and on studying their evolution under the stochastic 
dynamics of the quantum particle introduced in Sect.~\ref{models}. Key results concerning the evolution of coherent states are described in Sect.~\ref{sec:SMR}. 
In Sects.~\ref{sec:proof_estim}, ~\ref{sec:WM} and ~\ref{sec:ECS}, we present a number of technical details needed in the proofs of our main results, and we indentify a family of coherent states invariant under the stochastic evolution introduced in Sect.~\ref{models}. An important technical result that guarantees the existence of a family of coherent states invariant under this evolution is established in Sect.~\ref{sec:WM}.\\
Mathematical details concerning the time evolution of coherent states are presented in Sect.~\ref{sec:ECS}. 
Specific examples of (quasi-free) particle dynamics are discussed in Sect.~\ref{sec:examples}. \\
In three appendices we present a number of techncial details, including proofs of some auxiliary results stated in 
Sects.~\ref{models} and~\ref{sec:SMR}. 

\section{More about models, and summary of main results}\label{models}
Before entering into specifics it may be helpful to present a short guide to this section.
We begin this section by studying the stochastic dynamics of a (quasi-) freely 
moving quantum particle subjected to repeated measurements of its approximate position; see Eqs.~\eqref{history} 
(Subsect.~\ref{mod}) and \eqref{evolution}, \eqref{prob}, below. As a key ingredient of our analysis, we then
identify a family of coherent states (with a specific ``squeezing parameter matrix'', $\widehat{W}$, see \eqref{hatWeq} 
and \eqref{solu-W}) 
that turns out to be \textit{invariant} under the time-reversed stochastic dynamics of the particle; see Eq.~\eqref{invariance}. In 
Eqs.~\eqref{process} and \eqref{Born}, we introduce a stochastic process with values in the classical phase 
space $\Gamma$ of the particle that indexes the trajectory of coherent states occupied by the particle under the 
forward dynamics. Our first main result, Theorem \ref{thm:Q_and_AR}, 
relates the sequence of measurement data of approximate particle positions to the sequence of 
phase space points determined by the 
stochastic process in Eq.~\eqref{process}. In Theorem~\ref{estimation}, we determine the best guess of the initial 
condition of a phase space trajectory of the stochastic process introduced in \eqref{process} from its tail. 
Theorem~\ref{thm:AR} is an auxiliary result. Our last result is Theorem~\ref{thm:POVM_equivalence}, which 
relates the  positive operator-valued measure (POVM) induced by sequences of approximate particle position 
measurements to a POVM taking values in the 
space of coherent states. \\

Next, we describe the models studied in this paper more precisely.
The Hilbert space of pure state vectors of a quantum particle 
has been introduced in Eq. \eqref{Hilbert} and is given by
\begin{equation}\label{space}
\mathcal{H}:= L^{2}(\mathbb{R}^{d}, \d^{d}x)\,.
\end{equation}
The algebra of bounded operators on $\mathcal{H}$ is denoted by $B(\mathcal{H})$. The position- and momentum 
operators of the particle  are given by 
\begin{equation}\label{X-P}
{\bf{X}}:= \big(X_1, \dots, X_{d} \big) \quad \text{and } \quad {\bf{P}}:= \big( P_1,\dots, P_{d} \big)
\end{equation}
satisfying the usual Heisenberg commutation relations
\begin{equation}\label{CR}
\big[X_j, X_k \big]=0, \quad \big[P_j, P_k \big]=0, \quad  \big[X_j, P_k \big] = i \hbar \delta_{jk}\,,\,\,\,\forall j, k\,.
\end{equation}
Henceforth we will set $\hbar=1$. 

In our idealized model, measurements of the approximate position of the particle are described by a 
positive-operator-valued measure (POVM) determined by the operators
\begin{equation}\label{POVM1}
V_{\bf{q}}\equiv V_{\mathbf{q}}(\mathbf{X})  := \frac{1}{\big[(2\pi)^{d} \text{det}\Sigma\big]^{\frac{1}{4}}} \text{exp}\Big\{-\frac{1}{4}({\bf{X}}-{\bf{q}}) \Sigma^{-1}({\bf{X}}- {\bf{q}})^{t} \Big\}  \,,
\end{equation}
where $\Sigma$ is a positive-definite $d\times d$ matrix whose square root encodes the precision, previously denoted by $\lambda > 0$, of the position measurement; (see \eqref{Gauss}, Subsect. \ref{mod}).
We note that 
\begin{equation}\label{POVM}
\int \d^{d}{q} \, V_{\bf{q}}^{*}\cdot V_{\bf{q}} = 1\,,
\end{equation}
where $\d{\bf{q}} \equiv \d^{d}q$ denotes Lebesgue measure on $\mathbb{R}^{d}$. Let $\rho$ be a general, possibly mixed state of the particle; i.e., $\rho$ is a \textit{density matrix} (a non-negative, trace-class operator, 
with tr$(\rho)=1$) on $\mathcal{H}$. Let ${\bf{Q}}=\big(Q_1, \dots, Q_d \big)$ be a vector-valued random variable whose values, 
$\mathbf{q} \in \mathbb{R}^{d}$, indicate
the result of an indirect measurement of the approximate position of the particle. If the state of the particle is given by the density matrix
$\rho$ then the law of ${\bf{Q}}$ is given by $\rho(V_{\bf{q}}^{*}\cdot V_{\bf{q}})\,\d^{d}q$, where $\rho(A):= \text{tr}(\rho\,A)$,
for $A\in B(\mathcal{H})$.

\begin{proposition}\label{prop:law_Q_0}
	Let $\mu_\rho$ be the spectral measure of the commuting operators $\mathbf X=\big(X_1,\dots,X_d\big)$ in 
	the state $\rho$. 
	Let $\mathbf{Z}$ be an $\mathbb R^d$-valued Gaussian random variable with mean $0$ and covariance $\Sigma$ independent of $\mathbf{X}$. Then,
	$$\mathbf Q\sim \mathbf X+\mathbf Z.$$
Assuming that $\rho$ is such that $\rho(\mathbf{XX}^t)$ exists, we have that
\begin{equation}
\label{3}
\mathbb E(\mathbf Q) = \rho(\mathbf X)\quad \text{and} \quad \mathbb{E}(Q_i \cdot Q_j)=\Sigma_{ij} +
 \rho(X_i\cdot X_j)\,.
\end{equation}
\end{proposition}
This proposition follows immediately from the definitions.\\
To describe the effect of an instantaneous measurement of the approximate position of the particle on its state we 
follow the conventional wisdom of quantum mechanics: In the course of such a measurement whose result is given 
by some vector $\mathbf{q}\in \mathbb{R}^{d}$, the state $\rho$ of the particle changes according to the rule

\begin{equation}\label{meas}
\rho \mapsto \big[ \text{tr}(\rho \,V_{\bf{q}}^{*}\, V_{\bf{q}})\big]^{-1} V_{\bf{q}}\, \rho\, V_{\bf{q}}^{*} .
\end{equation}
After every measurement of the approximate position of the particle its state evolves unitarily for one time step. If the 
inital state  is given by $\rho$ then the state after one time step is given by $U\rho\,U^{*}$, where $U$ is the unitary propagator of the particle corresponding to one time step. As described in the introduction, we assume that the particle evolves quasi-freely. Thus, let $\Gamma:= \mathbb{R}_{x}^{d}\oplus \mathbb{R}^{d}_{p}$ denote the classical phase space of the particle, and let $S: \Gamma \rightarrow \Gamma$ be a symplectic $2d \times 2d$ matrix encoding a 
classical evolution of the particle in one time step. Then $S$ determines a unitary operator, $U\equiv U_S$, on $\mathcal{H}$ with the property that
\begin{equation}\label{2}
U_{S}^{*}\,\big( \mathbf{X}, \mathbf{P} \big)\, U_{S}= \big(\mathbf{X}, \mathbf{P}\big)\cdot S^{t}
\end{equation}
where a slight abuse of notation has been committed: The operator on the left side is defined by
$$ U_S^{*} \big(\mathbf{X}, \mathbf{P} \big) U_{S}:= \Big(U_S^{*} X_1 U_S,\dots, U_S^{*} X_d U_S, 
U_S^{*} P_1 U_S, \dots, U_S^{*} P_d U_S\Big)\,,$$
and the matirx $S$ on the right side of \eqref{2} acts as the identity on $\mathcal{H}$. Note that 
$$U_S^{*}=U_S^{-1}=U_{S^{-1}}\,.$$
We suppose that the instantaneous measurement of the approximate particle position is always done at the beginning of each 
time step. The projection of a point $(\mathbf{X}, \mathbf{P})\in \Gamma$ onto the subspace 
$\mathbb{R}^{d}_{\mathbf{x}}\oplus{\mathbf{0}}$ is denoted by 
\begin{equation}\label{proj}
\big[\big(\mathbf{X}, \mathbf{P}\big)\big]:= \mathbf{X}\,. 
\end{equation}
We define the operator
\begin{align}\label{evolution}
W_{n}(\underline{\mathbf{q}}_{n})&:= U_S V_{\mathbf{q}_n}(\mathbf{X}) \dots U_S V_{\mathbf{q}_0}(\mathbf{X})\nonumber\\
 &\,\,= U_{S^{n+1}} \prod_{j=0}^{n} V_{\mathbf{q}_j}(\big[\big(\mathbf{X}, \mathbf{P}\big) (S^{t})^{j} \big])\,.
\end{align}
This operator encodes the effective dynamics of a particle in $n$ time steps when its approximate position 
is measured at the beginning of each time step and it then evolves for one time step according to the propagator $U_S$, 
before its approximate position is measured again.
A probability measure, $\d\mathbb{P}_{\rho}^{(n)}$, on $\big(\mathbb{R}^{d}\big)^{\times{n+1}}$ is introduced by setting
\begin{align}\label{prob}
\d \mathbb{P}^{(n)}_{\rho}(\mathbf{q}_0,\dots, \mathbf{q}_n)&:= \rho\big( W_{n}(\underline{\mathbf{q}}_n)^{*}\cdot 
W_{n}(\underline{\mathbf{q}}_n)\big) \d\underline{\mathbf{q}}_{n}\nonumber\\
&\,\,= \rho\Big(\big[ \prod_{j=0}^{n} V_{\mathbf{q}_j}(\big[\big(\mathbf{X}, \mathbf{P}\big) (S^{t})^{j} \big])\big]^{*}\cdot  \prod_{j=0}^{n} V_{\mathbf{q}_j}(\big[\big(\mathbf{X}, \mathbf{P}\big) (S^{t})^{j} \big])  \Big)\d\underline{\mathbf{q}}_n
\end{align}
where $\d\underline{\mathbf{q}}_{n} =\prod_{j=0}^{n} \d\mathbf{q}_{j}$ is Lebesgue measure on $\big(\mathbb{R}^{d}\big)^{\times{n+1}}$.
This measure determines the law of the random variable $\underline{\mathbf{Q}}_{n}:=\big(\mathbf{Q}_0,\dots \mathbf{Q}_{n}\big)$ with values $\underline{\mathbf{q}}_{n}:= \big(\mathbf{q}_0, \dots, \mathbf{q}_n\big)$ in 
$(\mathbb{R}^{d})^{\times (n+1)}$, assuming that the initial state of the particle is given by $\rho$.

The family  $\big\{\d\mathbb{P}_{\rho}^{(n)}\vert n=0,1,2,\dots\big\}$ of probability measures is consistent. Hence, by Kolomogorov's extension theorem,  there exists a unique probability measure, $\d\mathbb{P}_{\rho}$, on the space $\mathfrak{Q}=\big(\mathbb{R}^{d}\big)^{\mathbb{N}_0}$ whose marginal on the first $n+1$ arguments is given by $\d\mathbb{P}_{\rho}^{(n)}$, for all $n=0,1,2,\dots$

\subsection{A first glance at the main result}\label{main-result}
Recall that the classical motion of the particle in one time step is decribed by a symplectic $2d\times 2d$ matrix $S$ on its phase space $\Gamma= \mathbb{R}^{d}_{x}\oplus \mathbb{R}^{d}_{p}$. It is convenient to introduce the notation 
\begin{equation}\label{ABCD}
S= \begin{pmatrix} S_{xx} & S_{xp}\\
                              S_{px} & S_{pp} \end{pmatrix}\,,
\end{equation}                              
where $S_{xx}, S_{xp}, S_{px}$ and $S_{pp}$ are $d\times d$ matrices (i.e., elements of 
$\mathbb{M}_{d\times d}(\mathbb{R})$).  
We require the following assumption.

\begin{assumption}{AW}{ass:W} 
The $d\times d$ matrix $S_{xp}$ is invertible.
\end{assumption}

We consider the matrix equation
\begin{equation}\label{hatWeq}
W\cdot S_{xx}-S_{px}= \big(S_{pp}-W\cdot S_{xp}\big)\big( W- \frac{i}{2} \Sigma^{-1}\big)\,,
\end{equation}       
for an unknown $d\times d$ matrix $W$, where $\Sigma$ is the matrix introduced  in \eqref{POVM1}.
In Sect.~5, we show that Assumption {\bf{AW}} implies that this equation has a solution, $W=\widehat{W}$, with the following properties:
\begin{equation}\label{solu-W}
 \widehat{W}=\widehat{W}^{t}, \quad \Im{\widehat{W}} >0, \quad \text{and}\quad \big(2 \Im{\widehat{W}}\big)^{-1} < \Sigma\,.
\end{equation}
\noindent
{\bf{Remark}}: \textit{A typical example for which Assumption {\bf{AW}} does not hold is an infinitely heavy particle 
corresponding to $S=\mathbf{1}$. Then the position and velocity operators are commuting conservation laws, and the 
effect of repeated measurements of the approximate particle position (namely ``purification'') can be analyzed with the 
methods described in \cite{BCFFS}.} \\

Let $\vert W, \zeta \rangle, \zeta \in \Gamma,$ be a coherent state in $\mathcal{H}$ localized near the point $\zeta$ 
in the phase space $\Gamma$ of the particle, and with squeezing parameter matrix $W$; see Eq. \eqref{coh_eq}, 
Subsect. \ref{coh-states}, below. Let $\widehat{W}$ be the solution \eqref{solu-W} of Eq. \eqref{hatWeq}. Then we have that
\begin{equation}\label{invariance}
 \vert \widehat{W}, \zeta \rangle:= \frac{W_{n}(\underline{\mathbf{q}}_{n})^{*}\ket{\widehat W,\zeta_{n+1}}}{\Vert W_{n}(\underline{\mathbf{q}}_{n})^{*}\ket{\widehat W,\zeta_{n+1}} \Vert }\,, 
\end{equation}
is again a coherent state with the same squeezing parameter matrix $\widehat{W}$, where 
$W_{n}(\underline{\mathbf{q}}_{n})$ has been defined in \eqref{evolution}, for arbitrary $n=0,1,2, \dots$. 
In other words, the squeezing parameter matrix $\widehat{W}$ of a coherent state, 
$\vert \widehat{W}, \zeta_{n+1} \rangle$, is \textit{invariant} under the (time-reversed) stochastic evolution described in 
\eqref{invariance}. This invariance property plays a key role in our analysis. Its proof consists in carrying out a number of 
explicit calculations; see Lemmas~\ref{SLemma} and~\ref{qLemma},
Sect.~\ref{sec:ECS}. Given a phase space point $\zeta$, we must find out how to choose the point $\zeta_{n+1}$ 
such that Eq. \eqref{invariance} holds, i.e., we must determine the law of the stochastic process 
$\big(\zeta \equiv \zeta_0, \dots, \zeta_n, \dots \big)$ determined by \eqref{invariance}. This process is studied next.\\

\noindent
{\bf{Definition of a stochastic process}}.\\
Let $\widehat{W}\in \mathbb{M}_{d\times d}(\mathbb{C})$ be as in \eqref{solu-W}. We define a 
$2d \times d$ matrix $K$ mapping $\mathbb{R}^{d}_{x}$ into $\Gamma$ by setting
\begin{equation}\label{K-matrix}
K:= \begin{pmatrix} {\bf{1}}\\ \Re \widehat{W} \end{pmatrix} \big(2\Im \widehat{W} \big)^{-1} 
\big(\Sigma - (2 \Im \widehat{W})^{-1}\big)^{-1}\,.
\end{equation}
In the following, points in $\Gamma$ are denoted by $\zeta = (\xi, \pi)$\footnote{In the following $(\xi, \pi)$ and $\begin{pmatrix} \xi \\ \pi \end{pmatrix}$ are usually taken to denote the same phase-space point $\zeta$.}, 
and the projection of $\zeta$ onto 
\mbox{$\mathbb{R}^{d}_{x} \oplus \{0\}\subset \Gamma$} is denoted by $\big[\zeta\big]\equiv \big[(\xi, \pi)\big]=\xi$ (see Eq. \eqref{proj}). We introduce a \textit{stochastic process} $\big(\zeta_n\big)_{n\in \mathbb{N}_0},\, \zeta_n \in \Gamma, \, \forall n \in \mathbb{N}_0$, as follows:
\begin{equation}\label{process}
\zeta_{n+1}= S\big(\zeta_n - K\eta_n\big), \quad n=0,1,2, \dots\,, 
\end{equation}
where $\zeta_{0}:= \zeta\,,$ and $(\eta_n)_{n\in \mathbb{N}_0}$ are i.i.d. Gaussian random vectors 
in $\mathbb{R}^{d}_{x}$ \textit{independent} of 
$\zeta$ with mean $0$ and covariance given by the matrix $\Sigma - (2 \Im \widehat{W})^{-1}$, (which, 
by \eqref{solu-W}, is positive-definite).\\
The law of the initial random variable $\zeta$ is given by the probability measure
\begin{equation}\label{Born}
\rho\big(\vert \widehat{W}, \zeta\rangle \langle \widehat{W}, \zeta \vert \big)\d\lambda(\zeta)= \langle \widehat{W}, \zeta \vert\, \rho\, \vert \widehat{W}, \zeta\rangle \d\lambda(\zeta)\,, 
\end{equation}
where $\d\lambda(\zeta):=\frac{\d^{2d}\zeta}{(2\pi)^{d}}$ is the Liouville measure on $\Gamma$; i.e., it is determined by the expectation of the density matrix $\rho$  in the coherent state 
$\vert \widehat{W}, \zeta \rangle$ localized near the phase space point $\zeta$ and with squeezing parameter 
matrix given by $\widehat{W}$. Coherent states and their properties will be described more precisely in the next subsection. 

A key result in our analysis is to identify the law, $\d\mathbb{P}_{\rho}$, on the space $\mathfrak{Q}$ of random 
sequences $\underline{\mathbf{Q}}_{\infty}$, i.e., of measurement records of approximate particle positions, 
with the law on sequences of noisy positions centered on the discrete trajectory of a classical particle with 
dynamics given by by the matrix $S$. In particular, we claim that the averages of 
the measured approximate particle positions follow a deterministic particle trajectory determined
by $S$.

\begin{theorem}\label{thm:Q_and_AR}
We require assumption {\bf{AW}}. Let $\rho$ be a density matrix on $\mathcal{H}$, and let $d\mathbb{P}_{\rho}$ be the law of the sequence, $\underline{\mathbf{Q}}_{\infty}$, of random variables whose values, $\underline{\mathbf{q}}_{\infty}$, are the results of indirect measurements of the approximate positions of the particle at the beginning of every time step. Let $\big(\zeta_n \big)_{n\in \mathbb{N}_0}$ be the stochastic process introduced in Eq. \eqref{process}, above. Then the following ``equality in law'' holds:
$$\underline{\mathbf{Q}}_{\infty} = \Big( \xi_n + \eta_n \Big)_{n\in \mathbb{N}_0}, \quad \text{where}\quad  \xi_n:=\big[\zeta_{n} \big] \,.$$
If the density matrix $\rho$ is such that the expectation values $\rho(\vert \mathbf{X} \vert)$ and $\rho(\vert \mathbf{P}\vert)$ exist then 
$$\mathbb{E}_{\rho}(\mathbf{Q}_n) = \big[ \big(\rho(\mathbf{X})\,, \rho(\mathbf{P})\big)(S^{t})^{n} \big]\,,$$
where
$$\mathbb{E}_{\rho}(F):=\int \d\mathbb{P}_{\rho}(\underline{\bf{q}}_{\infty}) F(\underline{\mathbf{q}}_{\infty})\,.$$
\end{theorem}     

We sketch the proof of this Theorem, deferring all technical details to later sections. According to Lemma~\ref{lem:stable_W},
$\Sigma-(2\Im\widehat W)^{-1}$ is a real, positive-definite $d\times d$ matrix, and, according to Proposition~\ref{prop:characterization_coherent_POVM}, we have that 
$\mathbb E(\zeta)=(\rho(\mathbf{X}),\rho(\mathbf{P}))$. 
Let $\d\mathbb Q_\rho^{(n)}$ be the joint law of $\zeta \equiv \zeta_0$ and $\eta_0,\dotsc,\eta_n$. By definition,
$$ \d \mathbb Q_\rho^{(n)}(\zeta, \eta_0,\dots, \eta_n)=\langle\widehat{W},\zeta \vert\,\rho\, \vert \widehat{W},\zeta \rangle \,
\d\lambda(\zeta) \prod_{k=0}^n \d \mathcal N(\eta_k,\Sigma-(2\Im\widehat W)^{-1})\,,$$
where $\d \mathcal{N}\big(\eta, \Delta\big)$ is the Gaussian measure on $\mathbb{R}^{d}$ with mean $0$ and covariance 
$\Delta \in \mathbb{M}_{d}(\mathbb{R})$. 
It then follows from Lemma~\ref{lem:P_as_convec_comb} that 
$$\d\mathbb{P}_\rho(\underline{\mathbf{q}}_n) = \int_{\Gamma} 
\frac{\d\mathbb Q_\rho^{(n)}(\zeta, \mathbf q_0 -\xi_0,\dots,\mathbf{q}_n-\xi_n)}{\d^{2d}\zeta} \, \d^{2d}\zeta\,,$$
for arbitrary $n\in \mathbb{N}_0$. (The reader should notice that if, in Lemma~\ref{lem:P_as_convec_comb},  
$\mathbf{q}_n - \xi_n$ is replaced by $\eta_n$ then the sequence $\big( \zeta \big)_{n \in \mathbb{N}_0}$ 
appearing in Lemma~\ref{lem:P_as_convec_comb} is identical to the one defined in \eqref{process}.)
 This proves the theorem.  \hspace{8.5cm}$\square$

\subsection{Digression on coherent states}\label{coh-states}
\textit{Coherent states}, 
$$\vert \widehat{W}, \zeta \rangle \in \mathcal{H}\,,  \text{ with  }\,
\zeta := \begin{pmatrix} \xi \\ \pi \end{pmatrix} \in \Gamma =\mathbb{R}^{2d} ,$$
 turn out to play a key role in our analysis. For, first, up to normalization, the image of a coherent state under the action of the operators $W_{n}(\underline{\mathbf{q}}_n)$ and $W_{n}(\underline{\mathbf{q}}_n)^{*}$ defined in \eqref{evolution} is again a coherent state, for arbitrary $n\in \mathbb{N}_0$; and, second, coherent states will enable us to derive an explicit expression for the measure $\d\mathbb{P}_{\rho}$ in terms of a Gaussian Markov process.

Given a symmetric matrix $W \in M_{d \times d}(\mathbb{C})$ with a positive definite imaginary part, i.e.,
\begin{equation}\label{squeeze}
W^t = W\quad \text{ and} \quad \mathrm{Im}(W):=(W -W^*)/(2 \i) > 0\,,
\end{equation}
and a point $\zeta\in \Gamma$, we define the coherent state, $\vert W, \zeta \rangle$, as the unit vector in the Hilbert space $\mathcal{H}$ solving the equation

\begin{equation}\label{coh_eq}
\begin{pmatrix}
-W & {\bf{1}} \end{pmatrix}
\begin{pmatrix}
\mathbf{X}  -\xi \\
\mathbf{P}  -\pi \end{pmatrix} \vert W, \zeta \rangle = 0\,.
\end{equation}
The matrix $W$ is called the \textit{squeezing parameter (matrix)}. We are using the following notational conventions: For a vector $\mathbf{Z} \in \mathbb{R}^{2d}$ and a vector $\psi \in \mathcal{H}$, we
define 
$$\mathbf{Z} \psi:= \begin{pmatrix} Z_1 \psi \\ \vdots \\ Z_{2d} \psi \end{pmatrix} \in \mathbb{R}^{2d} \otimes \mathcal{H}\,,$$
and, in Eq.~\eqref{coh_eq}, $\mathbf{Z} := \begin{pmatrix} \mathbf{X}-\xi \\ \mathbf{P}- \pi \end{pmatrix}$ and 
$\psi:= \vert W, \zeta\rangle$. For an arbitrary $k\times 2d$ matrix $A$, we set
$$A \mathbf{Z} \psi:= \begin{pmatrix} \sum_{j=1}^{2d} A_{1j} Z_j \psi \\ \vdots \\ \sum_{j=1}^{2d} A_{kj} Z_j \psi \end{pmatrix} \in \mathbb{C}^{k} \otimes \mathcal{H}\,.$$
The solution of Eq.~\eqref{coh_eq} is unique, up to a phase. In position space, the wave function, $\phi(x \vert W, \zeta)$, corresponding to the coherent state $\vert W, \zeta \rangle, \zeta \in \Gamma,$ is given by
\begin{equation}\label{S_repr}
\phi(x \vert W, \zeta) = \gamma^{-\frac d4}\det(2\Im W)^{\frac{1}{4}} \exp\big[ \frac{\iu}{2} (x - \xi)^t W (x - \xi) + \iu x^t \pi\big]\,,
\end{equation}
where $\gamma$ is equal to twice the constant $\pi$ (not to be confused with the particle momentum $\pi$). Manifestly, 
$\xi$ and $\pi$ have the meaning of the average particle position and momentum, respectively, in the coherent state 
$\vert W, \zeta \rangle$. This can also be seen directly from (\ref{coh_eq}) by calculating the scalar product of the 
left hand side of (\ref{coh_eq}) with $\phi(\cdot \vert W, \zeta)\equiv \vert W, \zeta \rangle$.

Coherent states with an arbitrary, but fixed squeezing parameter matrix $W$, with properties as in \eqref{squeeze}, form a non-orthogonal partition of unity (or decomposition of the identity), namely
\begin{equation}
\label{coherent_identity}
\int_{\mathbb{R}^{2d}} \ket{W, \zeta}\bra{W, \zeta} \d\lambda(\zeta) = {\bf{1}}\,,
\end{equation}
 see, e.g., \cite[Theorem 8.85]{DerezinskiGerard}). Hence $\big(\Gamma, \ket{W, \zeta}\bra{W, \zeta}\,\d\lambda(\zeta)\big)$  defines a positive operator-valued measure (POVM). 
 
\begin{proposition}\label{prop:characterization_coherent_POVM}
	Let $\rho$ be a density matrix on $\mathcal{H}$, and let $\zeta$ be a $\Gamma$-valued random variable whose 
	law is given by $\rho(\ket{ W,\zeta}\bra{ W,\zeta})d\lambda(\zeta)$, where $W$ is symmetric, with a 
	positive-definite imaginary part. Then we have that
	$$\xi=\mathbf{X} +\mathbf Z_x\,, \quad \text{and}  \quad \pi=\mathbf{P} + \mathbf Z_p\,,$$
where $\mathbf{X}$ ($\mathbf{P}$, resp.) is a random variable whose law is given by the spectral measure of the operator 
$\mathbf X$ (the operator $\mathbf P$, resp.) in the state $\rho$, and $\mathbf Z_x$ ($\mathbf Z_p$, resp.) is an independent 
Gaussian random vector with mean $0$ and covariance $(2\Im W)^{-1}$ ($2\Im W + 2\Re W(\Im W)^{-1}\Re W$, resp.).
\end{proposition}
The proof of this proposition is given in Appendix \ref{app:CS}.\\

As mentioned in Subsect.~\ref{main-result}, coherent states with squeezing parameter matrix $\widehat{W}$ solving Eq. \eqref{hatWeq} play a particularly important role in this paper. For, using Lemmas~\ref{SLemma} and~\ref{qLemma}, one observes that the set of coherent states with squeezing parameter matrix $\widehat W$ is invariant under the mapping 
$$\ket{\widehat W,\zeta}\mapsto \frac{W_{n}(\underline{\mathbf{q}}_{n})^{*}\ket{\widehat W,\zeta}}{\Vert W_{n}(\underline{\mathbf{q}}_{n})^{*}\ket{\widehat W,\zeta} \Vert }\,,$$
 where the operators 
$W_{n}(\underline{\mathbf{q}}_{n})$ are defined in Eq. \eqref{evolution}.

Throughout the paper, the symbol $\widehat W$ always refers to the solution of~\eqref{hatWeq}, with properties as described in \eqref{solu-W}.

\subsection{A maximum-likelihood point of view}\label{max-likelyhood}

In this subsection, we change our point of view, as compared to the one adopted in Subsect. \ref{main-result}. We attempt to reconstruct the trajectory of states of the particle, in particular its initial state, from the sequence, $\underline{\mathbf q}_\infty$, of records of approximate position measurements.\\

In addition to Assumption {\bf{AW}} of Subsection \ref{main-result} we will henceforth require two further assumptions on the symplectic matrix $S$. 
\begin{assumption}{AS}{ass:S}
	All eigenvalues of $S$ have modulus $1$ (i.e., the classical dynamics does not exhibit an unstable manifold).
\end{assumption}

To state our last assumption on $S$ and our next result we introduce two matrices, $R \in M_{2d \times d}(\mathbb{R})$ and 
$M \in M_{2d \times 2 d}(\mathbb{R})$, as follows.
\begin{equation}\label{Ma}
R := \begin{pmatrix} {\bf{1}}_{d}\\ \Re \widehat W \end{pmatrix}\big(2\Im \widehat W \big)^{-1}\Sigma^{-1}\quad \text{  and  } 
\quad M := \left({\bf{1}}_{2d} - \begin{pmatrix}\multirow{2}{*}{R} & 0 \\ & 0 \end{pmatrix} \right) S^{-1}.
\end{equation}
Let $K$ be the matrix defined in \eqref{K-matrix}. Then we have that
$$K=S^{-1}M^{-1}R.$$
Our last assumption on $S$ is the following one.
\begin{assumption}{AM}{ass:M}
All eigenvalues of the matrix $M$ defined in \eqref{Ma} have modulus strictly less \mbox{than $1$.} 
\end{assumption}
\noindent
{\bf{Remark}}: Assumption {\bf{AM}} is technical; (it may be a consequence of Assumptions {\bf{AW}} and {\bf{AS}}, but we have not pursued this question). In Section~\ref{sec:examples}, we will show that all our assumptions on $S$ hold in the examples of a freely moving particle, of a harmonic oscillator, and of a particle moving in a constant external magnetic field.  \\
Notice that if $S={\bf{1}}_{2d}$, i.e., for an infinitely heavy particle, $1$ belongs to the spectrum of the matrix $M$, and Assumption~\ref{ass:M} does not hold.\\

The next theorem provides an explicit expression for the maximum-likelihood estimator of the initial state of the particle, given the sequence, $\underline{\mathbf q}_\infty$, of records of approximate position measurements, which turns out to be a coherent state with squeezing parameter matrix $\widehat{W}$.
The proof of this result relies on Assumptions {\bf{AW}}, {\bf{AS}} and {\bf{AM}} in an essential way; (it is given in Section~\ref{sec:proof_estim}).

\begin{theorem}\label{estimation} 
If assumptions {\bf{AW}}, {\bf{AS}} and {\bf{AM}} hold then the following series 
\begin{equation}\label{MLE}
\widehat{\zeta} = \left( \begin{array}{c}
\widehat{\xi} \\
\widehat{\pi}
\end{array}
\right) = \sum_{j=0}^\infty M^{j} R \,\mathbf{q}_j
\end{equation}
converges almost surely with respect to $\d\mathbb{P}_\rho$ and defines a random vector $\widehat{\zeta}\in \Gamma$.

The coherent state $\ket{\widehat{W}, \widehat{\zeta}}$ is the maximum likelihood estimator 
of the initial state of the particle. Given the initial state $\rho$, the probability density of finding $\ket{\widehat{W}, \widehat{\zeta}}$
as the most likely initial state  of the particle is given by Born's Rule, namely by 
$$\langle \widehat{W}, \widehat{\zeta}\vert\,\rho\, \vert \widehat{W}, \widehat{\zeta}\rangle\,.$$

\end{theorem}

Next, we consider the process of shifts of $\widehat{\zeta}$. For an arbitrary $n\in \mathbb N_0$, we set
\begin{equation}\label{MLEn}
\widehat{\zeta}_n = \left( \begin{array}{c}
\widehat{\xi}_n \\
\widehat{\pi}_n
\end{array}
\right) = \sum_{j=0}^\infty M^{j} R\, \mathbf q_{j+n}.
\end{equation}

Even though $\widehat{\zeta}_n$ seems to depend on the entire future after time $n$, it actually turns out that
$(\widehat{\zeta}_n)_{n\in \mathbb N_0}$ is a Markov process.

\begin{theorem}
	\label{thm:AR}
	If Assumptions {\bf{AW}}, {\bf{AS}} and {\bf{AM}} hold
	then the process $(\widehat{\zeta}_n)_{n\in \mathbb N_0}$ has the same law as the process $(\zeta_n)_{n\in\mathbb N_0}$ defined in Eq. \eqref{process}. In particular, $\mathbb{E}_\rho[\widehat{\zeta}_n| \widehat{\zeta}_0] = 
	S^n \widehat{\zeta}_0$.
	
	The sequence $(\mathbf Q_n-\widehat{\xi}_n)_{n\in\mathbb N_0}$ consists of i.i.d. Gaussian random variables with mean zero and covariance $\Sigma-(2\Im\widehat W)^{-1}$ (identical to the random variables 
	$(\eta_n)_{n\in \mathbb{N}_0}$ introduced in \eqref{process}), and
	$$
	\mathbb{E}_\rho \big(\mathbf Q_n  | \widehat{\zeta}_0\big) =  \big[S^n \widehat{\zeta}_0\big]\,,
	$$
	with $\big[\zeta\big] := \xi$\,, for $\zeta=\begin{pmatrix} \xi\\ \pi \end{pmatrix}$.
\end{theorem}
\begin{proof}
	This theorem is an immediate consequence of Lemma~\ref{lem:AC}, which asserts that the random variables 
	$\widehat{\zeta}_n$ defined in \eqref{MLEn} and ${\zeta}_n$ defined in \eqref{process} have identical laws. 	\end{proof}
\subsection{What POVM's have got to do with it}\label{povm}
We conclude our survey of the main results established in this paper by describing what certain \textit{positive operator-valued measures} (POVM) have to do with analyzing the stochastic dynamics \eqref{evolution} of the particle.

We recall that a POVM on a Hilbert space $\mathcal{H}$ is a couple $(\Omega,F)$ consisting 
of a measurable space, $\Omega$, and a map, $F$, from the measurable sets of $\Omega$ to the set of 
non-negative bounded operators acting on $\mathcal{H}$, with the properties that $F(\Omega)=\id$,
and that $\rho(F(\d x))$ is a probability measure on $\Omega$, for any density matrix $\rho$ on 
$\mathcal{H}$.

To state our result we introduce a notion of convergence for POVM's.

\begin{definition}[POVM convergence]
A sequence, $\big((\Omega,F_n)\big)_{n\in \mathbb{N}_0}$, of POVM's is said to converge to a POVM 
$(\Omega,F)$, in a given sense iff for any state $\rho$ the sequence of probability measures $(\rho(F_n(\d x)))_{n\in\mathbb N_0}$ converges to $\rho(F(\d x))$ in the same sense.
\end{definition}
One example of convergence relevant for our results is total variation convergence:
$$\lim_{n\to\infty}\sup_{A}|\rho(F_n(A))-\rho(F(A))|=0$$
for any state $\rho$ with the supremum taken over the measurable sets of $\Omega$.

We also define a notion of POVM mapping.
\begin{definition}[Image POVM]
	Given a POVM $(\Omega_1,F)$ and measurable map $f$ from $\Omega_1$ to a measurable space $\Omega_2$, the image POVM of $(\Omega_1,F)$ by $f$ is $(\Omega_2,F\circ f^{-1})$.
\end{definition}
Remark that the measure $\rho(F\circ f^{-1}(\d y))$ is the image measure of $\rho(F(\d x))$ by $f$.

The following theorem is proven in Section~\ref{sec:proof_estim}.
\begin{theorem}\label{thm:POVM_equivalence}
	We assume that {\bf{AW}}, {\bf{AS}} and {\bf{AM}} hold (see Sects.~\ref{main-result} and \ref{max-likelyhood}). Let 
	$\d\mathbb{P}_{\tau}$ be as in \eqref{prob} (with $\rho$ replaced by $\tau$), and let $\widehat{\zeta}$ be as in \eqref{MLE}. 
	Then, for any {\bf{strictly positive}} density matrix $\tau$ with the property
	 that $\tau(\mathbf X^t \mathbf X)+\tau(\mathbf P^t\mathbf P)<\infty$, the limit,
	$$\lim_{n\to\infty}\left(\mathfrak Q,\frac{W_n(\underline{\mathbf q}_n)^*W_n(\underline{\mathbf q}_n)}{\tau(W_n(\underline{\mathbf q}_n)^*W_n(\underline{\mathbf q}_n))}\d\mathbb{P}_\tau\right)=
	\left(\mathfrak Q,\frac{\ket{\widehat W, \widehat{\zeta}}\bra{\widehat W, 
	\widehat{\zeta}}}{\langle \widehat W, \widehat \zeta\vert\, \tau\, \vert\widehat W, \widehat \zeta\rangle}\d\mathbb{P}_\tau\right) \,,$$
	exists in total variation. Moreover, the image POVM of the limit POVM by $\widehat{\zeta}$ is the coherent state POVM $(\Gamma,\ket{\widehat{W},\zeta}\bra{\widehat{W},\zeta}\d\lambda(\zeta))$.
\end{theorem}

\section{The law of data recording approximate particle positions}\label{sec:SMR}
In this section, we present an explicit expression for the law of the random sequences, $\underline{\mathbf Q}_\infty$, of approximate position measurements, i.e., for the measure, $\d \mathbb P_\rho$, on the space, $\mathfrak{Q}$, 
of sequences, $\underline{\bf{q}}_{\infty}$, of approximate particle positions. Our expression serves 
to complete the proof of Theorem~\ref{thm:Q_and_AR}. 

We recall the definition~\eqref{Ma} of the matrices $R$ and $M$:
\begin{equation*}
R =  \begin{pmatrix} {\bf{1}}_{d}\\ \Re\widehat W \end{pmatrix}\big(2\Im\widehat W\big)^{-1}\Sigma^{-1}\quad \text{  and  } 
\quad M = \left({\bf{1}}_{2d} - \begin{pmatrix}\multirow{2}{*}{R} & 0\\ & 0 \end{pmatrix} \right) S^{-1}.
\end{equation*}
To state the main result of this section we must consider several 
equivalent definitions of the process $\big(\zeta_n \big)_{n=0,1,2,\dots}$ introduced in \eqref{process}.

\begin{lemma}\label{lem:rec}
	We require assumption {\bf{AW}}. Let \mbox{$\underline{\mathbf q}_{n}\in (\mathbb R^d)^{n+1}$} be an arbitrary, but fixed measurement record consisting of an $(n+1)$-tuple of approximate particle positions, and let $\zeta\in\mathbb R^{2d}$
	be an arbitrary, but fixed phase-space point.\\
	Then the time evolution of coherent states with squeezing parameter matrix $\widehat{W}$, as defined in Eqs. \eqref{hatWeq} and \eqref{solu-W}, is given by 
	\begin{equation}\label{eq:rec_bwd_coherent}
	V_{\mathbf q_k}^*U_S^*\ket{\widehat W,\zeta_{k+1}}\propto\det M^{-\frac12}\sqrt{\mathcal N(\mathbf{q}_k-\xi_k, \Sigma-(2\Im\widehat W)^{-1})}\ket{\widehat W,\zeta_k}\,,
	\end{equation}
for all $k\in \{0,\dotsc,n\}$. Here $\propto$ means ``equality up to a phase'', and $\mathbb{R}^{d} \ni \xi\mapsto\mathcal N(\xi,\Delta)$ is the denisty of a Gaussian probability measure on $\mathbb{R}^{d}$ with mean $0$ and covariance $\Delta$.
In this formula, the sequence, $(\zeta_k)_{k=0}^n$, of points in phase space is determined by the recursion	
\begin{equation}\label{eq:rec_fwd_zeta}
	\zeta_{k+1}=S(\zeta_k - K(\mathbf{q}_k-\xi_k))\,,\quad \text{for any  }\, k\in\{0,\dotsc,n\}\,, 
\end{equation}
	where $\zeta_0:= \zeta, \,\zeta_k=\begin{pmatrix} \xi_k \\ \pi_k \end{pmatrix}$, i.e., $\xi_k=\big[\zeta_k \big]$, and
	$$K=S^{-1}M^{-1}R=\begin{pmatrix}
	{\bf{1}}\\ \Re \widehat W \end{pmatrix}\big(2\Im \widehat W\big)^{-1}\Big(\Sigma-(2\Im\widehat W)^{-1}\Big)^{-1}\,,$$
	see \eqref{K-matrix} and \eqref{Ma}. For any $k\in\{1,\dots,n+1\}$,
	\begin{equation}\label{eq:explicit_zeta}
		\zeta_k=M^{-k}\zeta - \sum_{j=0}^{k-1} M^{-k+j}R\mathbf{q}_j\,,
	\end{equation}
	and, for $k\in \{0,\dotsc,n\}$,
	\begin{equation}\label{eq:rec_bwd_zeta}
		\zeta_{k}=S^{-1}\zeta_{k+1} +K(\mathbf{q}_{k}-\xi_{k})=S^{-1}\zeta_{k+1} +\begin{pmatrix}
		\multirow{2}{*}{R}&0\\&0
		\end{pmatrix}\left(\begin{pmatrix}
		\mathbf{q}_{k}\\0
		\end{pmatrix}-S^{-1}\zeta_{k+1}\right).
	\end{equation}			
\end{lemma}
\noindent
{\bf{Remark}}. The proof of this lemma consists in performing lengthy, but rather straightforward and explicit calculations. 
It is given in Appendix \ref{Pf}.

\begin{lemma}\label{lem:P_as_convec_comb}
	We suppose that assumption {\bf{AW}} holds. Given a point $\zeta$ in phase space $\Gamma$ and a sequence, 
	$\underline{\mathbf q}_{n}\in (\mathbb R^d)^{n+1}$, of approximate particle positions, we define a sequence, 
	$\Big\{\zeta_k= \begin{pmatrix} \xi_k\\ \pi_k \end{pmatrix} : k=0,1,\dots, n \Big\}$, of points in phase space by setting
	$$
	\zeta_{k+1}= S(\zeta_k - K(\mathbf{q}_k - \xi_k)), \quad k=0,1,\dots n, \,\text{ with }\, \, \zeta_0:= \zeta\,,$$
	where the matrix $K$ is defined as in Eq. \eqref{K-matrix}, Sect. \ref{main-result}.\\
	Then the following ``master formula'' holds:
\begin{align}\label{masterformula}
\frac{d\mathbb P_\rho^{(n)}(\underline{\mathbf q}_n)}{d\underline{\mathbf q}_n}&\equiv
      \rho(W_n^*(\underline{\mathbf q}_n)W_n(\underline{\mathbf q}_n))\nonumber \\
&=\int_{\mathbb R^{2d}} \prod_{k=0}^{n} \mathcal N(\mathbf{q}_k-\xi_k,\Sigma-(2\Im\widehat W)^{-1})\ \rho(\ket{\widehat W,\zeta}\bra{\widehat W,\zeta})d\lambda(\zeta)\,,
\end{align}
	 for any $n\in\mathbb N_0$. (Note that the matrix $\Sigma-(2\Im(\widehat W))^{-1}$ is real and 
	 positive-definite\textcolor{magenta}{).}
 \end{lemma}
\begin{proof}
Equation~\eqref{eq:explicit_zeta} in Lemma~\ref{lem:rec} implies that $\zeta_{n+1}=M^{-(n+1)}\zeta + \tilde \zeta_n$, where 
$\tilde\zeta_n$ is a point in phase space that is \textit{independent} of $\zeta$. Since the measure $\d\lambda$ is translation-invariant, and replacing $\zeta\to M^{n+1}\zeta$, it follows from Eq. \eqref{coherent_identity} (partition of unity) that
\begin{equation}\label{dec-of-1}
\det{M^{-(n+1)}}\int_{\mathbb R^{2d}} \ket{\widehat W,\zeta_{n+1}}\bra{\widehat W,\zeta_{n+1}}d\lambda(\zeta)={\bf{1}} \,.
\end{equation}
	By definition,
	$$\frac{d\mathbb P_\rho^{(n)}(\underline{\mathbf q}_n)}{d\underline{\mathbf q}_n}=\rho(W_n^*(\underline{\mathbf q}_n)W_n(\underline{\mathbf q}_n)).$$
Inserting the partition of unity in \eqref{dec-of-1} between $  W_n^*(\underline{\mathbf q}_n)$ and 
	$  W_n(\underline{\mathbf q}_n) $, and iterating Eq.~\eqref{eq:rec_bwd_coherent} of Lemma~\ref{lem:rec} $(n+1)$ times, we conclude that
	$$\rho(W_n^*(\underline{\mathbf q}_n)W_n(\underline{\mathbf q}_n))=\int_{\mathbb R^{2d}} \prod_{k=0}^{n}\mathcal N(\mathbf q_k- \xi_k,\Sigma-(2\Im \widehat W)^{-1}) \rho(\ket{\widehat W,\zeta}\bra{\widehat W,\zeta})\  \d\lambda(\zeta)\,.$$
This completes the proof of the lemma.
\end{proof}

\section{Parameter estimation and proofs of \mbox{Theorems \ref{estimation} and \ref{thm:POVM_equivalence}}}\label{sec:proof_estim}

In this section, we first show that if the initial state, $\rho$, of the particle has the property that 
\mbox{$\rho(\mathbf X^t\mathbf X)+\rho(\mathbf P^t\mathbf P)<\infty$} then the $L^2$-norms of 
the random variables $\mathbf Q_n$ and $\zeta_n$ 
are polynomially bounded in $n \in \mathbb{N}_0$. Here $\mathbf Q_n$ is the random variable corresponding to the 
$n^{th}$ measurement of the approximate particle position, and the phase space points $\zeta_n, n=0,1,2,\dots,$ are 
given by the process defined in Eq.~\eqref{process} and Theorem~\ref{thm:Q_and_AR}.

\begin{lemma}\label{lem:poly_bound_var}
We require assumptions {\bf{AW}} and {\bf{AS}}. Let $\tau$ be a {\bf{strictly positive}} density matrix on $\mathcal{H}$ with the property that \mbox{$\tau(\mathbf X^t\mathbf X)+\tau(\mathbf P^t\mathbf P)<\infty$.} Let $(\zeta_n)_{n\in\mathbb N_0}$ be the 
process defined in Eq.~\eqref{process}, with the law of the randomly chosen initial condition $\zeta_0$ 
in phase space given by 
$\zeta_0 \equiv \zeta \sim \tau(\ket{\widehat W,\zeta}\bra{\widehat W,\zeta})\d\lambda(\zeta)$. Let $\d \mathbb Q_\tau$ 
be the joint law of $\zeta_0$ and the sequence of random vectors $\big({\bf{Q}}_{n} = \xi_n+ \eta_n \big)_{n\in \mathbb N_0}$, 
with $\xi_n=\big[\zeta_n\big]$  (see Theorem~\ref{thm:Q_and_AR}). Then there exists a finite constant, $C$, such that, 
for any $n\in \mathbb N_0$,
	$$\mathbb E_\tau(\|\mathbf Q_n\|_2^2)\leq C(1+n^{4d-1})\quad \mbox{and} \quad \mathbb E_\tau(\|\zeta_n\|_2^2)\leq C(1+n^{4d-1}),$$
	where $\mathbb{E}_{\tau}$ denotes an expectation w.r. to the measure $\d \mathbb Q_\tau$.
\end{lemma}
\begin{proof}
	First, we estimate the $L^2$-norm of $\zeta_n, n=0,1,2,\dots$. By definition,
	$$\zeta_n= S^n\zeta_0-\sum_{k=0}^{n-1}S^{n-k}K\eta_k.$$
	By Proposition~\ref{prop:characterization_coherent_POVM}, 
	$\mathbb E_{\tau}(\|\zeta_0\|_2^2)=\tau(\mathbf X^t\mathbf X)+\tau(\mathbf P^t\mathbf P)<\infty$. Since $\zeta_0$ and 
	$\eta_k, k=0,1,2,\dots,$ are independent random variables, we have that
	\begin{align*}
	\mathbb E_{\tau} (\|\zeta_n\|_2^2)\leq\ &\mathbb E_{\tau}(\|\zeta_0\|_2^2)+\sum_{k=1}^{n}\Tr(S^{k}K(\Sigma- (2 \Im \widehat{W})^{-1}) K^t(S^t)^k)\\
	\leq\ &\mathbb E_{\tau}(\|\zeta_0\|_2^2)+\sum_{k=1}^{n}\Tr((S^t)^kS^{k}K(\Sigma- (2 \Im \widehat{W})^{-1}) K^t),
	\end{align*}
		where Tr denotes the trace on $\mathbb{M}_{2d\times 2d}(\mathbb{R})$.
	Using the inequality $\Tr(AB)\leq \|A\|\Tr(B)$, valid for arbitrary positive semi-definite matrices $A$ and $B$, we obtain that
	$$\mathbb E_{\tau} (\|\zeta_n\|_2^2)\leq\mathbb E_{\tau}(\|\zeta_0\|_2^2)+\Tr(K(\Sigma- (2 \Im \widehat{W})^{-1}) K^t)\sum_{k=1}^{n}\|S^k\|^{2}.$$
	Assumption~\ref{ass:S} and Proposition~\ref{prop:NB} (Appendix \ref{NB}) then imply that there exists a finite 
	constant $C'$ such that $\|S^n\|\leq C' n^{2d-1}$. Thus there exists a constant $C<\infty$ such that
	$$\mathbb E_{\tau} (\|\zeta_n\|_2^2)\leq C(1+n^{4d-1}).$$
	By Theorem~\ref{thm:Q_and_AR}, $\mathbf Q_n=\xi_n+\eta_n$, $\xi_n =\big[\zeta_n \big]$, where the random vector 
	$\eta_n \in \mathbb{R}^{d}$ is independent of $\zeta_n$, for each $n$. Hence, 
	$\mathbb E_{\tau}(\mathbf Q_n)=\mathbb E_{\tau}(\xi_n)$ and  
	$\operatorname{Var}\mathbf Q_n=\operatorname{Var}\xi_n +\Sigma-(2\Im\widehat W)^{-1}$. Consequently,
	$$\mathbb E_{\tau}(\|\mathbf Q_n\|_2^2)=\mathbb E_{\tau}(\|\xi_n\|_2^2)+\Tr(\Sigma-(2\Im \widehat W)^{-1})$$
	and the bound on $\mathbb E_{\tau}(\|\zeta_n\|_2^2)$ yields the lemma.
\end{proof}
Let the sequence of random variables $\big(\widehat{\zeta}_n\big)_{n\in \mathbb{N}_0}$ be as given by Eq. \eqref{MLEn}, i.e., 
$$\widehat{\zeta}_n = \left( \begin{array}{c}
\widehat{\xi}_n \\
\widehat{\pi}_n
\end{array}
\right) = \sum_{j=0}^\infty M^{j} R\, \mathbf q_{j+n}\,,$$
see Theorem \ref{thm:AR}, Sect. \ref{max-likelyhood}.
We propose to show that $\widehat \zeta_n$ is a good estimator of $\zeta_n$, in the sense that $\widehat{\zeta}_n=\zeta_n$, almost surely. Here we define the random variables $\mathbf Q_n$ as  $\mathbf Q_n:=\xi_n+\eta_n$, so that both processes are defined on the same probability space comprising $\zeta_0$  and $(\eta_n)_{n\in\mathbb N_0}$.
\begin{lemma}\label{lem:equality_estimator}
	We require assumptions  {\bf{AW}}, {\bf{AS}} and {\bf{AM}}. Let $\tau$ be a {\bf{strictly positive}} density matrix with the property that 
	 $\tau(\mathbf X^t\mathbf X)+\tau(\mathbf P^t\mathbf P)<\infty$. Let $(\zeta_n)_{n\in\mathbb N_0}$ be the process defined in Eq.~\eqref{process}, with $\zeta_0\sim \tau(\ket{\widehat W,\zeta}\bra{\widehat W,\zeta})\d\lambda(\zeta)$. Let 
	 $\d\mathbb Q_\tau$ be the joint law of $\zeta_0$ and the sequence $\underline{\mathbf{Q}}_{\infty}=(\xi_n+ \eta_n)_{n\in\mathbb N_0}$.
	 Then
	$$\widehat{\zeta}_n= \zeta_n\,, \quad \d\mathbb{Q}_{\tau}-\text{almost surely}\,,$$
	for any $n \in \mathbb N_0$.

	\end{lemma}

\begin{proof}
	We begin by proving that the process $\big(\widehat \zeta_n \big)_{n \in \mathbb{N}_0}$ introduced in Eq.~\eqref{MLEn} is well defined. From Lemma~\ref{lem:poly_bound_var}, the Markov inequality and the Borel--Cantelli lemma we infer that 
		$$\lim_{k\to\infty}\lambda^k\mathbf Q_k=0\quad\mbox{ and }\quad \lim_{k\to\infty}\lambda^k\zeta_k=0,\quad \d\mathbb Q_{\tau}-\text{ almost surely},$$
		for an arbitrary positive $\lambda<1$.
Assumption~\ref{ass:M} then ensures that $\widehat \zeta_n$ exists and is finite almost surely.
	Moreover, Eq.~\eqref{eq:explicit_zeta} in Lemma~\ref{lem:rec} implies that
	$$\zeta_n=M^{-n+k}\zeta_{k} - \sum_{j=n}^{k-1}M^{-n+j}R\,\mathbf q_{j}\,,$$
	for any natural number $k$.
	Eq. \eqref{MLEn} is then obtained by taking the limit $k\to\infty$, and using Assumption~\ref{ass:M}. This completes the proof of the lemma.
\end{proof}

Next, we show that, for an arbitrary density matrix $\rho$, the measure $\d\mathbb P_\rho$ is absolutely continuous with respect to a reference measure 
$\d\mathbb P_\tau$, where $\tau$ is a strictly positive density matrix.
\begin{lemma}\label{lem:AC}
	We require assumptions {\bf{AW}}, {\bf{AS}} and {\bf{AM}}. Let $\tau$ be a strictly positive density matrix such that $\tau(\mathbf X^t\mathbf X)+\tau(\mathbf P^t\mathbf P)<\infty$.\\
	Then $\d\mathbb P_\rho\ll \d\mathbb P_\tau$, for any density matrix $\rho$, and the Radon--Nikodym derivative, 
	$\d\mathbb P_\rho/\d\mathbb P_\tau$, is given by
	$$\frac{\d\mathbb P_\rho}{\d\mathbb P_\tau}=\frac{\rho(\ket{\widehat W,\widehat \zeta}\bra{\widehat W,\widehat \zeta})}{\tau(\ket{\widehat W,\widehat \zeta}\bra{\widehat W,\widehat \zeta})}\,.$$
	Moreover, if $\d\mathbb Q_\rho$ defines the same law as $\d\mathbb Q_\tau$, as defined in Lemma~\ref{lem:equality_estimator} (but with $\tau$ replaced by a density matrix $\rho$) then
		$$\widehat \zeta_n=\zeta_n\,, \quad \d\mathbb Q_\rho\mbox{-almost surely}\,, \quad \forall n\in \mathbb N_0\,.$$
\end{lemma}

\begin{proof}
	Let $(\zeta_n)_{n\in\mathbb N_0}$ be the process defined in Theorem~\ref{thm:Q_and_AR}, with 
	$\zeta_0\sim \tau(\ket{\widehat W,\zeta}\bra{\widehat W,\zeta})\d\lambda(\zeta)$. Let $\d \mathbb Q_\tau$ be the joint law of $\zeta_0$ and the sequence $\underline{\mathbf Q}_\infty=(\xi_n+ \eta_n)_{n\in\mathbb N_0}$.
	From Equation~\eqref{masterformula} in Lemma~\ref{lem:P_as_convec_comb} we infer that 
$$\mathbb E_{\rho}(f)=\int_{\mathfrak Q \times \Gamma} f(\underline{\mathbf q}_{\infty})\, \frac{\rho(\ket{\widehat W,\zeta_0}\bra{\widehat W,\zeta_0})}{\tau(\ket{\widehat W,\zeta_0}\bra{\widehat W,\zeta_0})}\d\mathbb Q_{\tau}(\underline{\mathbf q}_{\infty}, \zeta_0),$$
for an arbitrary bounded continuous function $f$ on $\mathfrak Q$.

	Since, according to Lemma~\ref{lem:equality_estimator}, $\widehat \zeta=\zeta_0$, $\d\mathbb Q_\tau$-almost surely, 
	$$\mathbb E_{\rho}(f)=\int_{\mathfrak Q \times \Gamma} f(\underline{\mathbf q}_{\infty})\, \frac{\rho(\ket{\widehat W,\widehat{\zeta}}\bra{\widehat W,\widehat{\zeta}})}{\tau(\ket{\widehat W,\widehat{\zeta}}\bra{\widehat W,\widehat{\zeta}})}\d\mathbb Q_{\tau}(\underline{\mathbf q}_{\infty}, \zeta_0),$$
and the first part of the lemma is proved. The second part follows from the absolute continuity of $\d\mathbb Q_\rho$ with respect to $\d\mathbb Q_\tau$.
\end{proof}

We now turn to the proof of the maximum likelihood estimation. The explicit expression of the Radon-Nikodym derivative obtained in Lemma \ref{lem:AC} induces an explicit expression of the maximum likelihood estimator of the initial particle state, given the full sequence, $\underline{\mathbf Q}_\infty$, of outcomes of approximate position measurements.
\begin{proof}[Proof of Theorem~\ref{estimation}]
	The existence of $\widehat \zeta$ is guaranteed by Lemma~\ref{lem:AC}. Let $\d\mathbb Q_\rho$ be the joint law of $\zeta_0$ and the sequence $\underline{\mathbf Q}_\infty=(\xi_n+ \eta_n)_{n\in\mathbb N_0}$.
	We recall that Theorem~\ref{thm:Q_and_AR} shows that $\int_{\mathbb R^{2d}}\d\zeta_0 \frac{\d\mathbb Q_\rho}{\d\zeta_0}(\cdot) =\d\mathbb P_\rho(\cdot)$, namely the marginal of $\d\mathbb Q_\rho$ on the sequence 
	$(\mathbf Q_n)_{n\in\mathbb N_0}$ is given by $\d\mathbb P_\rho$.
	
	Let $\tau$ be a strictly positive density matrix with 
	\mbox{$\tau(\mathbf X^t\mathbf X)+\tau(\mathbf P^t\mathbf P)<\infty$.}
	By Lemma~\ref{lem:AC},
	$$\frac{\d\mathbb P_\rho}{\d\mathbb P_\tau}=\frac{\rho(\ket{\widehat W,\widehat \zeta}\bra{\widehat W,\widehat \zeta})} {\tau(\ket{\widehat W,\widehat \zeta}\bra{\widehat W,\widehat \zeta})}\,.$$
	This holds for an arbitrary density matrix $\rho$. The maximum likelihood estimator of the initial state $\rho$, 
	given $\underline{\mathbf Q}_\infty$, is defined by
	$$\widehat \rho_{ML}=\operatorname{argmax}_{\rho} \frac{\d\mathbb P_\rho}{\d\mathbb P_\tau}\,,$$
	with the maximum taken over all possible states. It follows that,
	$$\widehat \rho_{ML}= \operatorname{argmax}_{\rho} \rho(\ket{\widehat W,\widehat \zeta}\bra{\widehat W,\widehat \zeta}).$$
	Hence, $\widehat \rho_{ML}=\ket{\widehat W,\widehat \zeta}\bra{\widehat W,\widehat \zeta}$. 
\end{proof}

Our last concern in this section is to prove that the POVM's defined by the operators $W_n(\underline{\mathbf q}_n), n=0,1,2,\dots,$ converge, as $n\rightarrow \infty$.
\begin{proof}[Proof of Theorem~\ref{thm:POVM_equivalence}]
	The fact that the image POVM of the limit POVM by $\widehat{\zeta}$ is the coherent state POVM $(\Gamma,\ket{\widehat{W},\zeta}\bra{\widehat{W},\zeta}\d\lambda(\zeta))$ follows from Lemma~\ref{lem:equality_estimator}. 
	
	It remains to prove the convergence in total variation. It is sufficient to prove that, for any two states $\rho$ and $\tau>0$, with $\tau(\mathbf{X}^t\mathbf{X})+\tau(\mathbf{P}^t\mathbf{P})<\infty$,
	$$\lim_{n\to\infty}\frac{\rho(W_n(\underline{\mathbf Q}_n)^*W_n(\underline{\mathbf Q}_n))}{\tau(W_n(\underline{\mathbf Q}_n)^*W_n(\underline{\mathbf Q}_n))}=\frac{\rho(\ket{\widehat W,\widehat \zeta}\bra{\widehat W,\widehat \zeta})}{\tau(\ket{\widehat W,\widehat \zeta}\bra{\widehat W,\widehat \zeta})},\quad \mbox{in }L^1(\d\mathbb P_\tau)\mbox{-norm}.$$
	Indeed, let 
	$$M_n:=\frac{\rho(W_n(\underline{\mathbf Q}_n)^*W_n(\underline{\mathbf Q}_n))}{\tau(W_n(\underline{\mathbf Q}_n)^*W_n(\underline{\mathbf Q}_n))}\quad\mbox{and}\quad N=\frac{\rho(\ket{\widehat W,\widehat \zeta}\bra{\widehat W,\widehat \zeta})}{\tau(\ket{\widehat W,\widehat \zeta}\bra{\widehat W,\widehat \zeta})}\,.$$
	Let $\d \mu_n$ be the probability measure $M_n\d\mathbb P_\tau$ and $\d \mu$ be the probability measure $N\d\mathbb P_{\tau}$.
	Then for any measurable set $A\subset\mathfrak Q$,
	$$|\mu_n(A)-\mu(A)|=\big|\mathbb E_\tau\big(\chi_{A}(M_n-N)\big)\big|\leq \mathbb E_\tau\big(\big|M_n-N\big|\big),$$
	with $\chi_A$ the characteristic function of the set $A$.
	Hence $L^1(\d\mathbb P_\tau)$ convergence of $M_n$ to $N$ implies total variation convergence of $\mu_n$ to $\mu$. We now prove the $L^1$ convergence. 
	
	By definition of the measures $\mathbb P_\rho^{(n)}$ and $\mathbb P_\tau^{(n)}$,
	$$d\mathbb P_\rho^{(n)}=M_n\d\mathbb P_\tau^{(n)}\,.$$
	Lemma~\ref{lem:AC} tells us that
	$$\d\mathbb P_\rho^{(n)}=\mathbb E_\tau\left(N|\underline{\mathbf Q}_n\right)\d\mathbb P_\tau^{(n)}.$$
	Hence $(M_n)$ is a closed martingale, with
	$$M_n=\mathbb E_\tau\left(N|\underline{\mathbf Q}_n\right).$$
	Consequently, from L\'evy's upwards theorem~\cite[Theorem 27.3]{JacodProtter},
	$$\lim_{n\to\infty}M_n=N,\quad
	\mbox{in } L^1(\d\mathbb P_\tau)\mbox{-norm}.$$
	This proves the convergence of the POVM's. 
\end{proof}

\section{Existence of a stable squeezing parameter matrix}\label{sec:WM}
In this section we prove that Eq.~\eqref{hatWeq} has a solution, $\widehat{W}$, with the desired properties whenever $S_{xp}$ is invertible. As a preliminary, we prove two lemmas concerning properties of certain special symmetric $d \times d$ matrices, $W$, that will turn out to be important to construct solutions of \eqref{hatWeq}.

\begin{lemma}
\label{lem:invertible}
Suppose that $W$ has positive-definite imaginary part. Then $S_{pp} - WS_{xp}$ is invertible.
\end{lemma}
\begin{proof}
Let $E$ and $F$ be two complex $d\times d$ matrices with the property that 
$$
P:=\frac{E F^* - F E^*}{-2 i}
$$
is positive-definite. Then $E$ and $F$ are invertible matrices. Indeed, suppose that $E$ is not invertible; then there exists a vector $z \in \mathbb{C}^{d}$ such that $E^*z = 0$, and hence $(z,Pz) = 0$, which contradicts the strict positivity of $P$. 
A similar argument proves that $F$ is invertible.

Next, let $(\tilde{E}, \tilde{F}) = (E,F) S$. Since $S$ is a symplectic matrix, we have that
$$SJS^t = J, \quad \text{ with }\,\, J=J =\left( \begin{array}{cc}
		0 & {\bf{1}} \\
		-{\bf{1}} & 0 
		\end{array} \right)\,,$$
 see Eq.~\eqref{s-1}. This implies that 
$$
\left( \tilde{E}, \tilde{F} \right) J \left( \begin{array}{r} 
						\tilde{E}^* \\
						\tilde{F}^*
					\end{array} \right) = \left( E, F \right) J \left( \begin{array}{r} 
						E^* \\
						F^*
					\end{array} \right) = EF^* - F E^*.
$$
Hence 
$$
\tilde{P}:=\frac{\tilde{E} \tilde{F}^* - \tilde{F} \tilde{E}^*}{-2 i}
$$
is positive, and  $\tilde{E}$ and $\tilde{F}$ are invertible.

The statement now follows by setting $E=-W$ and $F={\bf{1}}$, noticing that
$$
\left(-W, {\bf{1}}\right) S = (S_{px} - WS_{xx}, S_{pp} -WS_{xp}),
$$
and using our assumptions on $W$. 
\end{proof}

Next, we exhibit some stability of properties of $W$.
\begin{lemma}
\label{Wsymmetric}Assume that $W$ is symmetric, with a positive-definite imaginary part. Then 
$$\tilde W = (S_{pp} - W S_{xp})^{-1}(WS_{xx} - S_{px})$$ 
is symmetric and has a positive-definite imaginary part, too.
\end{lemma}
\begin{proof}
	We remark that
	\begin{equation}\label{3.1}
	(S_{pp} - W S_{xp})(S_{xx}^t W^t - S_{px}^t) = (W S_{xx} - S_{px})( S_{pp}^t - S_{xp}^t W^t) + W^t - W.
	\end{equation}
	This identity is a consequence of the symplectic nature of the matrix $S$, which is equivalent to the identities	
		\begin{equation}\label{symplec}
S_{xx} S_{xp}^t = S_{xp} S_{xx}^t, \quad S_{px} S_{pp}^t = S_{pp} S_{px}^t, 
\quad S_{xx} S_{pp}^t - S_{xp} S_{px}^t = {\bf{1}}\,,
\end{equation}
	(see Sect.~\ref{sec:examples}, Eq.~\eqref{ABCD_constr}). Using these identities again, we see that 
	\begin{align*}
	LHS &= - S_{pp} S_{px}^t + S_{pp} S_{xx}^t W^t + W S_{xp} 
	 S_{xp}^t  -  W S_{xp} S_{xx}^t W^t \\
	&=- S_{px}S_{pp}^t + (1+  S_{px}  S_{xp}^t) W^t + W(S_{xx}S_{pp}^t-1) - WS_{xx}S_{xp}^tW^t = RHS.
	\end{align*}
	Since $W$ is symmetric,~\eqref{3.1} implies that 
	\begin{equation*}
	\tilde W = (S_{pp} - W S_{xp})^{-1}(WS_{xx} - S_{px})= (S_{xx}^t W^t - S_{px}^t)(S_{pp}^t - S_{xp}^t W^t )^{-1} = 
	\tilde W^t.
	\end{equation*}
	Hence $\tilde W$ is symmetric, too.
	
	Using that the matrix elements of $S$ are real, identities~\eqref{symplec} can be written as 
	$$S_{xx} S_{xp}^* = S_{xp} S_{xx}^*, \quad S_{px} S_{pp}^* = S_{pp} S_{px}^* \,\,\, \text{  and  }\, \,\,
	S_{xx} S_{pp}^* - S_{xp} S_{px}^* = {\bf{1}}\,.$$ 
Using \eqref{3.1} one then finds that 
	\begin{align}\label{ping}
	(WS_{xx} -S_{px})(S_{pp}- WS_{xp})^* - (S_{pp}-WS_{xp})(WS_{xx} - S_{px})^* = W - W^*.
	\end{align}
	Multiplying this identity from the left by $(S_{pp} - WS_{xp})^{-1}$ and by 
	$([S_{pp} - WS_{xp}]^{*})^{-1}$ from the right, we obtain that 
	$$\Im\tilde W=(S_{pp}-WS_{xp})^{-1}\Im W {([S_{pp}-WS_{xp}]^{*})}^{-1}.$$
	Hence $\Im \tilde W>0$, and the lemma is proven.
\end{proof}

Finally we prove that a solution $\widehat W$ to~\eqref{hatWeq} exists with $\Im \widehat W>(2\Sigma)^{-1}$.
\begin{lemma}\label{lem:stable_W}
	We assume that $S_{xp}$ is invertible. Then there exists a solution, $\widehat{W}$, of Eq.~\eqref{hatWeq} with the following properties.
	 \begin{itemize}
	\item{$\widehat W$ is symmetric}
	\item{$\widehat{W}$ has a positive-definite imaginary part, and} 
	\item{$(2\Im \widehat W)^{-1}< \Sigma$.}
\end{itemize}
\end{lemma}
\begin{proof}
	Our proof relies on the Brouwer fixed-point theorem. We will show that $\widehat W$ is the fixed point of a continuous function that maps a compact convex set of $d\times d$ matrices into itself. Let 
	$$g(W):=(S_{pp}-WS_{xp})^{-1}(WS_{xx}-S_{px}).$$ 
	By Lemma~\ref{Wsymmetric}, $g$ maps the set of symmetric $d\times d$ matrices with positive-definite 
	imaginary part to itself. Let $\mathcal C$ be the set of symmetric matrices, $W$, such that 
	$\Im W\geq \frac12\Sigma^{-1}$. 
	Then, for any $W\in \mathcal C$,
	$$(S_{pp}-WS_{xp})^{-1}(WS_{xx}-S_{px})=S_{xp}^{-1}(D-W)^{-1}(WS_{xx}-S_{px}) ,$$
	where $D:=S_{pp}S_{xp}^{-1}$. We then note that Eq.~\eqref{ABCD_constr}, below, with $S^{-1} = - J S^t J$, instead of 
	$S$, implies that $D$ is real and symmetric. Using that $\Im W\geq \frac12\Sigma^{-1}$, one concludes that
	$$\|(D-W)^{-1}\|=\|(D-\Re W-\i \Im W)^{-1}\|\leq\|(\Im W)^{-1}\|\leq 2\|\Sigma\|.$$
	It follows that, on the one hand,
	$$\|(S_{pp} -WS_{xp})^{-1}(DS_{xx}-S_{px})\|\leq 2\|S_{xp}^{-1}\|\|\Sigma\|\|S_{pp}S_{xp}^{-1}S_{xx}-S_{px}\|$$
	and that, on the other hand, 
	$$\|(S_{pp} -WS_{xp})^{-1}(W-D)S_{xx}\|\leq \|S_{xp}^{-1}S_{xx}\|.$$
	Hence $C=\sup_{W\in\mathcal C}\|g(W)\|$ is finite.

	Let $f$ be the map from the space of $d\times d$ matrices to itself given by 
	$$f: W \mapsto g(W)+\frac\i2\Sigma^{-1}\,.$$
	Then $f$ maps every symmetric $d\times d$ matrix with positive-definite imaginary part to $\mathcal C$. 
	A matrix $W$ is a solution to~\eqref{hatWeq} if and only if $W=f(W)$.
	
	Let $\mathcal M$ be the set of matrices defined by
	$$\mathcal M:=\{W\in \mathcal C : \|W-\tfrac{\i}{2}\Sigma^{-1}\|\leq C\}=\mathcal C\cap B_{\frac12\Sigma^{-1}}(C).$$
	Here $B_X(r)$ is the closed ball in $\mathbb{M}_{d}(\mathbb{C})$ of radius $r$ centered at $X$. The set 
	$\mathcal{M}$ is convex, as it is the intersection of two convex sets, and it is compact, since $\mathcal C$ is closed 
	and the ball $B_{\frac12\Sigma^{-1}}(C)$ is compact. 
	
	Since $\mathcal M$ is a subset of $\mathcal C$, we have that, for any $W\in\mathcal M$, 
	$\|f(W)-\frac\i 2\Sigma^{-1}\|=\|g(W)\|\leq C$. 
	Moreover, since $\Im g(W)$ is positive-definite, by Lemma~\ref{Wsymmetric}, 
	$f(W)=g(W)+\tfrac\i2\Sigma^{-1}$ is such that $\Im f(W)\geq \frac12\Sigma^{-1}$. 
	Consequently, $f(\mathcal M)\subset\mathcal M$, and the existence of a matrix $\widehat W$ satisfying 
	$f(\widehat W)=\widehat W$ follows from the Brouwer fixed-point theorem. 
	
	To show that $\Sigma - (2\Im\widehat W)^{-1}$ is strictly positive, it suffices to apply the map $f$ to $\widehat{W}$ 
	once. Indeed, $2\Im \widehat W-\Sigma^{-1}=2\Im g(\widehat W)>0$, by Lemma~\ref{Wsymmetric}.
	This completes our proof.
	\end{proof}

\section{Time evolution of coherent states}\label{sec:ECS}
In this section we describe the action of $U_S^*$ and $V_{\mathbf q}$ on the set of coherent states. We start with the former and use the notation introduced in (\ref{ABCD}). In the following, $\propto$ stands for ``equality up to a phase''. We recall that we write $\zeta = (\xi, \pi)^t$, for points $\zeta$ in phase space $\Gamma = \mathbb{R}^{2d}$ .
\begin{lemma}
\label{SLemma}
For a coherent state $\vert W, \zeta \rangle$, 
where $W$ is a symmetric $d \times d$ matrix with a positive-definite
 imaginary part, we have that 
\begin{align*}
U_{S^{-1}} \vert W, \zeta \rangle & \propto \vert \tilde{W}, \tilde{\zeta} \rangle\,, \qquad \text{where} \qquad \tilde{\zeta} = S^{-1} \zeta, \\
\tilde{W} &= (S_{pp} - W S_{xp})^{-1} ( WS_{xx} - S_{px})\,,
\end{align*}
and $\tilde W$ is a symmetric $d \times d$ matrix with a positive-definite imaginary part.
\end{lemma}
\begin{proof}
We set $\vert W, \zeta \rangle \equiv \phi(W, \zeta) =: \phi$. We first note that, by Lemma~\ref{lem:invertible}, $\tilde W$ 
is well defined. The fact that $ \tilde W $ is symmetric with a positive-definite imaginary part follows 
from Lemma~\ref{Wsymmetric}.

Second, since $\phi$ is the unit vector solving Eq.~\eqref{coh_eq} in Sect.~\ref{coh-states}, which is unique up to a phase, and since $U_{S^{-1}}=U_S^*$, the vector $\tilde\phi:=U_{S^{-1}}\phi$ solves the equation 
$$ \begin{pmatrix} -W & {\bf{1}} \end{pmatrix} 
\left( \begin{pmatrix} \mathbf X\\\mathbf P \end{pmatrix} -\zeta\right)U_S\,\tilde\phi=0\,,$$
with $\zeta \equiv \zeta\cdot \mathbf{1}$, or, equivalently,
$$\begin{pmatrix} -W & {\bf{1}} \end{pmatrix} \left( U_S^*\begin{pmatrix} \mathbf X\\\mathbf P \end{pmatrix}U_S -\zeta \right)\tilde\phi=0\,.$$
Then~\eqref{2} implies that
$$\begin{pmatrix} -W & {\bf{1}} \end{pmatrix} S\left( \begin{pmatrix} \mathbf X\\\mathbf P \end{pmatrix} 
-S^{-1}\zeta\right) \tilde\phi=0.$$
By Lemma~\ref{lem:invertible}, $S_{pp} - WS_{px} $ is invertible. Hence 
$$ \begin{pmatrix} -W & {\bf{1}} \end{pmatrix} S = (S_{pp} - W S_{px}) \begin{pmatrix} -\tilde W  & {\bf{1}}\end{pmatrix}.$$ 
Consequently,
$$ \begin{pmatrix} -\tilde W  & {\bf{1}}\end{pmatrix} \left(\begin{pmatrix} \mathbf X\\\mathbf P \end{pmatrix}
 -S^{-1}\zeta\right)\tilde\phi = 0$$
and~\eqref{coh_eq} yields the lemma.
\end{proof}

The action of the operator $V_{\mathbf q}$ on coherent states is described in the following lemma.
\begin{lemma}
\label{qLemma}
For any vector $\mathbf q\in\mathbb R^d$ and a coherent state $\ket{W, \zeta}$, we have that
$$V_{\mathbf q} \ket{W, \zeta} \propto \|V_{\mathbf q} \ket{W, \zeta}\| \cdot \ket{\tilde{W}, \tilde{\zeta}}\,,$$ 
where
\begin{align*}
\tilde{\zeta}& :=  \zeta + \left(\begin{array}{c} \id \\ \Re W  \end{array} \right) (2\Im W)^{-1} \Xi^{-1}(\mathbf q-\xi), \quad \text{and }\\
 \tilde{W} &:= W + \frac{\iu}{2}\Sigma^{-1},
\end{align*}
with\,\, $\Xi = \Sigma + (2 \Im W)^{-1}$. Furthermore,
$$
\|V_{\mathbf q} \ket{W, \zeta}\|^2 = \mathcal{N}(\mathbf q - \xi, \Xi).
$$
\end{lemma}
\begin{proof}
Clearly, $V_{\mathbf q} \ket{W, \zeta} \in L^2(\mathbb R^d)$. It thus suffices to check  that (see \eqref{coh_eq})
\begin{equation}
\label{aux2}
(\mathbf P - \tilde{\pi} - \tilde{W}(\mathbf X - \tilde{\xi})) \, V_{\mathbf q}\, \phi = 0,
\end{equation}
where $\phi=\ket{ W, \zeta}$ and $\tilde\zeta=(\tilde\xi,\tilde\pi)$.

Using the commutation relation $$\mathbf P\, V_{\mathbf q} = V_{\mathbf q}\, (\mathbf P + \frac{\i}{2}\Sigma^{-1}(\mathbf X - \mathbf q))\,,$$ 
one sees that (\ref{aux2}) is equivalent to 
$$
\Big(\mathbf P + \frac{\i}{2}\Sigma^{-1}(\mathbf X -\mathbf q) - \tilde{\pi} - \tilde{W} (\mathbf X - \tilde{\xi})\Big)\, \phi =  0.
$$
Inserting the definitions of the quantities with tildes, the left hand side is seen to be equal to
$$
\Big(\mathbf P - \pi - W(\mathbf X- \mathbf q) + (\tilde{W}\, \Sigma\,\Xi^{-1}+ \Re W(2\Im W)^{-1}\,\Xi^{-1}) (\xi - \mathbf q)\Big) \,\phi.
$$
A direct computation shows that $ \tilde{W}\, \Sigma\, \Xi^{-1}+ \Re W\,(2\Im W)^{-1}\,\Xi^{-1} = W$, and hence (\ref{aux2}) follows from (\ref{coh_eq}).

The last equation follows from Proposition~\ref{prop:law_Q_0} and by noticing that the law of $\mathbf X$ determined by the coherent state $\ket{W,\zeta}$ is given by the Gaussian $\mathcal N(\xi,(2\Im W)^{-1})$, or from (\ref{S_repr}) and 
\begin{equation*}
\mathcal{N}(x, \Sigma_1) \mathcal{N}(y, \Sigma_2) = \mathcal{N}(x-y, \Sigma_1 + \Sigma_2) \mathcal{N}(z, \Sigma)\,,
\end{equation*}
where 
$
\Sigma = (\Sigma_1^{-1} + \Sigma_2^{-1})^{-1},\ z = \Sigma (\Sigma_1^{-1} x + \Sigma_2^{-1}y).
$
\end{proof}

\section{Examples of quasi-free particle dynamics} \label{sec:examples}
For simplicity, we henceforth consider the special family of POVM's introduced in \eqref{Gauss}, setting
$ \Sigma = \lambda^{2} \mathbf{1}\,.$ 
We recall Eq.  \eqref{hatWeq}, namely 
\begin{equation}
\label{hatWeqprimaprima}
WS_{xx} - S_{px} = (S_{pp} -  W S_{xp})( W-\tfrac{\i}{2} \lambda^{-2} \mathbf{1}).
\end{equation}
In Sect. \ref{sec:WM}, we have shown that this equation has a unique solution, $\widehat{W}$, that 
has positive imaginary part. In Eq.~\eqref{Ma}, Subsect. \ref{max-likelyhood}, we have introduced two matrices 
$R \in M_{2d \times d}(\mathbb{R})$ and 
$M \in M_{2d \times 2 d}(\mathbb{R})$, given by
\begin{align}\label{Todo}
R &= \begin{pmatrix} {\bf{1}}_{d}\\ \Re\widehat W \end{pmatrix}\big(2\Im\widehat W\big)^{-1}\lambda^{-2}\,,
\quad \text{  and  } \nonumber \\
\quad M &= \left({\bf{1}}_{2d} - \begin{pmatrix}\multirow{2}{*}{R} & 0 \\ & 0 \end{pmatrix} \right) S^{-1} =  
\left( \begin{array}{cc}
\kappa & 0 \\
- \Re(\widehat{W}) (1 - \kappa) & {\bf{1}}
\end{array} \right) S^{-1},
\end{align}
where 
\begin{equation}\label{kappa}
\kappa := {\bf{1}} - \frac{ \lambda^{-2}}{2 \Im \widehat{W}} \in \mathbb{M}_{d\times d}(\mathbb{R})
\end{equation}
We recall that the matrix $K$, defined in \eqref{K-matrix} of Subsect. \ref{main-result}, is given by
$$K=S^{-1}M^{-1}R\,,$$
see Lemma \ref{lem:rec}. Next, we verify assumption {\bf{AM}}, stated in Subsect. \ref{max-likelyhood}, for $d=1$.

\begin{lemma}\label{LLL1}
	Assume that $d =1$ and that $S_{xp} \neq 0$. Then the spectrum of the matrix $M$ is contained in the open unit disk.
\end{lemma}
\begin{proof}
	Taking the imaginary part of (\ref{hatWeqprimaprima}), with $W=\widehat{W}$, we get
	\begin{align}\label{pu1}
	\Im(\widehat W)S_{xx} &+  S_{pp} \frac{1}{2 \lambda^2} - \Re(\widehat W) S_{xp}
	\frac{1}{2 \lambda^2} \nonumber \\
	&= S_{pp} \Im(\widehat W) - \Re(\widehat W) S_{xp} \Im(\widehat W) -
	\Im(\widehat W) S_{xp} \Re(\widehat W).
	\end{align}
	For $d = 1$, $\widehat W,$ the blocks of $S$ and $\kappa$ are complex numbers, and Eqs. \eqref{pu1} and \eqref{kappa} 
	yield
	\begin{align}\label{pu2}
	S_{xp} \Re(\widehat{W}) = - \frac{S_{xx} - \kappa S_{pp}}{1 + \kappa}.
	\end{align}
	It then follows from  (\ref{Todo}) and the fact that det($S)=1$ that
	\begin{align}\label{pu3}
	\det{M} = \kappa, \quad \tr(M) = \kappa S_{pp} + S_{xx} + \Re(\widehat{W}) (1 - \kappa) S_{xp}.
	\end{align}
	With Eq.~\eqref{pu2} this implies that
	$$
	\tr(M) = 2\, \tr(S) \frac{\kappa}{1+ \kappa}.
	$$
	The eigenvalues, $\mu$, of $M$ satisfy the secular equation
	$$
	\mu^2 - \tr(M) \mu + \det(M) = 0.
	$$
 If a quadratic equation, $\lambda^2 + b \lambda + c = 0$, has real coefficients and $b^2 \leq  4 c$ then a solution 
	$\lambda=\mu$ satisfies $|\mu|^2=c$. By assumption {\bf{AS}}, the spectrum of $S$ is 
	contained in the unit circle. With det($S)=1$, this implies that
	$ |\tr(S)| \leq 2. $
	Hence 
	$$
	\tr(M)^2 \leq 16 \frac{\kappa^2}{(1+ \kappa)^2}.
	$$
	Since $4 \kappa \leq (1 + \kappa)^2$, and using \eqref{pu3}, we find that 
	$$\tr(M)^2 \leq 4 \kappa =  4 \det(M)\,.$$ 
	As noticed above, this implies that
	$$
\vert\mu\vert^2=\kappa.
	$$
	By Lemma~\ref{lem:stable_W}, $\kappa $ is strictly positive, and, by \eqref{kappa}, $\kappa  < 1$. Hence $\vert \kappa \vert <1$, and the proof is complete.
\end{proof}
This lemma has the following obvious corollary.\\

\noindent
{\bf{Corollary 7.2.}}\label{LLL2} \textit{Suppose that  $S_{k\ell} =  s_{k\ell} {\bf{1}} $, where $s_{k\ell} \in \mathbb{R} $, for every $k,\ell \in \{ x, p\} $. Then the spectrum of  $M $ lies in the open unit disk. } \\

\noindent
{\bf{Remarks about quasi-free particle dynamics.}} 
We consider a one-parameter group, $\big(S_{\tau}\big)_{ \tau\in \mathbb{R}}$, of $2d\times 2d$-matrices mapping \mbox{$\Gamma \equiv \mathbb{R}^{d}_{x}\oplus \mathbb{R}^{d}_{p}$} onto itself, with $S_{\tau} \rightarrow \mathbf{1}$, as $t\rightarrow 0$. Then $S_{\tau}= \text{exp}[\tau\, L]\,,$ for some $2d\times 2d$ matrix $L$. Let
$$J:=\begin{pmatrix} {\,\bf{0}}_{d}& {\bf{1}}_{d}\\ -{\bf{1}}_{d} & {\bf{0}}_{d} \end{pmatrix}$$
The matrices $S_{\tau}, \tau \in \mathbb{R},$ are symplectic iff
\begin{equation}\label{s-1}
S_{\tau} \,J \,(S_{\tau})^{t} =J\,.
\end{equation}
Let $$S_{\tau}\equiv S = \begin{pmatrix} S_{xx}&S_{xp}\\S_{px}&S_{pp} \end{pmatrix}$$
Then \eqref{s-1} is equivalent to
\begin{equation}\label{ABCD_constr}
S_{xx} S_{xp}^t = S_{xp} S_{xx}^t, \quad S_{px} S_{pp}^t = S_{pp} S_{px}^t, \quad S_{xx} S_{pp}^t - S_{xp} S_{px}^t = {\bf{1}}\,.
\end{equation}
By differentiating identity \eqref{s-1} in $\tau$ and setting $\tau=0$, and using that $S_{\tau}\rightarrow \mathbf{1}\,, \text{ as }\, \tau \rightarrow 0$, one shows that 
\begin{equation}\label{Symplec}
S_{\tau} \,\, \text{ is symplectic }\,, \forall \tau \in \mathbb{R}\,, \qquad \Leftrightarrow \qquad L\,J + J\,L^{t}=0
\end{equation}
Let $h(\xi, \pi)$ be a Hamilton function on $\Gamma$ that is \textit{quadratic} in Darboux coordinates, $\xi$ (position)  and 
$\pi$ (momentum). Then the Hamiltonian equations of motion corresponding to $h$ have solutions, $(\xi_{\tau}, \pi_{\tau})$, given by
\begin{equation}\label{ham}
\begin{pmatrix} \xi_{\tau}\\ \pi_{\tau} \end{pmatrix}= S_{\tau} \begin{pmatrix} \xi\\ \pi \end{pmatrix},
\end{equation}
where $\xi=\xi_0$ and $\pi=\pi_0$ are the initial conditions and $\big\{S_{\tau}=e^{\tau L} \vert \tau \in \mathbb{R}\big\}$ is a one-parameter group of symplectic matrices on $\Gamma$ determined by $h$. Explicitly, 
\begin{equation}\label{Poisson}
\big\{ h, \zeta_j \big\} = \sum_{k=1}^{2d} L_{jk} \zeta_{k}\,,
\end{equation}
where $\zeta = \begin{pmatrix} \zeta_1 \\ \vdots \\ \zeta_{2d}\end{pmatrix}  = \begin{pmatrix} \xi \\ \pi \end{pmatrix} \in \Gamma$, and 
$\big\{\cdot, \cdot \big\}$ denotes the Poisson bracket on $\Gamma$.

We define a Hamilton operator, $H(\mathbf{X}, \mathbf{P})$, acting on $\mathcal{H}$ by replacing $\xi$ by the position operator $\mathbf{X}$ and $\pi$ by the momentum operator $\mathbf{P}$ in the expression for $h$. Then we have that
\begin{equation}\label{comm}
i \Big[ H, \begin{pmatrix} \mathbf{X} \\ \mathbf{P} \end{pmatrix} \Big] = 
L \begin{pmatrix} \mathbf{X} \\ \mathbf{P} \end{pmatrix}\,,
\end{equation}
where $\big[\cdot, \cdot \big]$ denotes the commutator, and the solution of the Heisenberg equations of motion for the position- and momentum operators are given by
\begin{equation}\label{Heisenberg-dyn}
\begin{pmatrix} \mathbf{X}_{\tau} \\ \mathbf{P}_{\tau} \end{pmatrix} = S_{\tau} \begin{pmatrix} \mathbf{X} \\ \mathbf{P} \end{pmatrix}\,. 
\end{equation}
The symplectic nature of $S_{\tau}$ -- see Eqs. \eqref{s-1} and \eqref{ABCD_constr} -- guarantees that the commutation relations of the components of $\mathbf{X}_{\tau}$ 
and $\mathbf{P}_{\tau}$ are identical to the Heisenberg commutation relations of the corresponding components of 
$\mathbf{X}$ and $\mathbf{P}$; i.e.,
$$
\big[(\mathbf X_ \tau)_j, (\mathbf P_ \tau)_k] = \iu \delta_{jk}, \hspace{1cm} \big[(\mathbf X_ \tau)_j, (\mathbf X_ \tau)_k] = 0, \hspace{1cm} \big[(\mathbf P_ \tau)_j, (\mathbf P_ \tau)_k] = 0  .
$$
Von Neumann's uniqueness theorem then implies that the algebra 
$^{*}$-isomorphism determined by $S_{\tau}$ is implemented by a \textit{unitary} operator $U_{\tau}\equiv U_{S_{\tau}}$, where 
$\big(U_{\tau}\big)_{\tau \in \mathbb{R}}$ is a one-parameter unitary group on $\mathcal{H}$.

The matrix $S$ used in the bulk of the paper corresponds to $S=S_{\tau=1}$, and the remark just made yields Eq.~\eqref{2} of Sect. \ref{models}.\\

Next, we consider some specific examples.

\begin{enumerate}
	\item{{\bf{Freely moving particle}} \\
		In this example,
		\begin{equation}\label{freeham}
		H_{\text{free}}:= \frac{1}{2M} \mathbf{P}^{2}
		\end{equation}
		Then 
		$$L=\begin{pmatrix} \mathbf{0}& \frac{1}{M}\mathbf{1}\\ \mathbf{0} & \mathbf{0}\end{pmatrix}, \quad \text{ and  }\,\quad 
		S=\begin{pmatrix} \mathbf{1} & \frac{1}{M} \mathbf{1}\\ \mathbf{0} & \,\,\mathbf{1} \end{pmatrix} $$
		It is immediate to verify \eqref{Symplec} and assumptions {\bf{AW}} and {\bf{AS}}. Assumption {\bf{AM}} follows from Corollary 7.2.
	}
	\item{ {\bf{Harmonic oscillator}} \\
		The Hamiltonian of the he harmonic oscillator is 
		\begin{align}\label{osci}
		H_{{\rm ho}} :=  \frac{\omega}{2}\big[\mathbf P^2 +  \mathbf X^2\big],
		\end{align}
		One easily verifies that 
		$$L=\begin{pmatrix} {\bf{0}} & \omega \mathbf{1}\\ -\omega \mathbf{1} & \mathbf{0}\end{pmatrix}, \quad \text{  and }\quad 
		S=\begin{pmatrix} \,\, \text{cos}(\omega)\, \mathbf{1} &\text{sin}(\omega)\, \mathbf{1}\\- \text{sin}(\omega)\, \mathbf{1} & 
		\text{cos}(\omega)\, \mathbf{1} \end{pmatrix} $$
		Again, it is immediate to verify \eqref{Symplec} and assumptions {\bf{AW}} and {\bf{AS}}; and, as above, assumption {\bf{AM}} follows from Corollary 7.2.
	}
	\item{\bf{ Particle in a constant magnetic field, d=2}}\\
		We consider a particle moving in a plane perpendicular to the direction of a constant external magnetic field, 
		$\vec{B}= B\,\vec{e}_{z}$, which we take to be parallel to the $z$-axis. The Hamiltonian is given by
		\begin{equation}\label{magn-field}
		H_{b} :=  \frac{1}{2M} \big[ \Pi_1 ^2 + \Pi_2 ^2 \big]\,,
		\end{equation}
		where 
		\begin{equation}\label{velocities}
		\Pi_1 := P_1 - \frac{B}{2} X_2 \,,\qquad \Pi_2 := P_2 + \frac{B}{2} X_1 
		\end{equation}
		are the components of the operator $M\mathbf{V}$, with $\mathbf{V}$ the (gauge-invariant) velocity operator.
		We then have that
		\begin{equation}\label{magn-transl}
		\big[\Pi_1 , \Pi_2 \big] = -i B \cdot \mathbf{1}
		\end{equation}
		We introduce two further operators, the so-called guiding center operators,
		\begin{equation}\label{guiding-centers}
		W_1 := \Pi_1 + B X_2\,, \qquad W_2 := \Pi_2 - B X_1\,.
		\end{equation}
		One readily verifies that 
		\begin{equation}\label{W}
		\big[W_1 , W_2 \big] = iB\cdot \mathbf{1}, \qquad \big[\Pi_j, W_k \big] =0\,,\,\, \forall j, k =1,2\,.
		\end{equation}
		The operators $\Pi_j$ and $W_j$ are the quantizations of canonical (Darboux) coordinates $\varpi_j$ and $w_j\,, j=1,2,$ on $\Gamma$. Apparently, the operators $W_1$ and $W_2$ are conservation laws; for, the Hamiltonian commutes with both of them. It is equivalent to a harmonic oscillator Hamiltonian in the canonicaly conjugated operators $\Pi_1$ and 
		$\Pi_2$.\\
		In the variables $\Pi_1$ and $\Pi_2$, this example is equivalent to example 2. If, however, the conservation laws $W_1$ and $W_2$ are included then assumption {\bf{AW}} fails. In the coordinates compatible with the approximate position measurement, it turns out that our assumptions hold.\\
		
		In the operators $\big(X_1, X_2, P_1, P_2 \big)$, the generator, $L$, of the one-parameter group $\big(S_{\tau} \big)_{\tau \in \mathbb{R}}$ is given by
		$$L= \begin{pmatrix} 0&-\beta & M^{-1}&0\\
		\beta & 0 & 0 & M^{-1}\\
		-\beta^{2}M & 0 & 0 &-\beta\\
		0 & -\beta^{2}M & \beta & 0 \end{pmatrix} \,,$$
		where $\beta:=\frac{B}{2M}$, and it is straightforward to verify Eq.~\eqref{Symplec}.
		By exponentiation we find that the matrix
			$
			S = \exp(L)
			$
			has components
			$$
			S_{xx} = S_{pp} = \begin{pmatrix} \cos(\beta)^2 & -\cos(\beta) \sin(\beta) \\
			\cos(\beta)\sin(\beta) & \cos(\beta)^2 \end{pmatrix},
			$$
			and
			$$
			-(M \beta)^{-1} S_{px} = M \beta S_{xp} = \begin{pmatrix} \cos(\beta) \sin(\beta) & -\sin(\beta)^2 \\
			\sin(\beta)^2 & \cos(\beta)\sin(\beta) \end{pmatrix}.
			$$
			It follows that {\bf AW} holds whenever $\beta\notin\mathbb Z \pi$. Moreover, remarking that $S=\widehat S\otimes R(\beta)$ with $R(\beta)$ the $2\times2$ rotation matrix of angle $\beta$ and $\widehat{S}$ the $2\times2$ simplectic matrix given by
				$$\widehat{S}=\begin{pmatrix}
				\cos(\beta) & \frac{\sin(\beta)}{M\beta}\\-M\beta\sin(\beta) &\cos(\beta)
				\end{pmatrix},$$
			an explicit computation of the spectrum of $S$ shows that {\bf AS} holds.
		
			We can furthermore show that eq.~\eqref{hatWeqprimaprima} has a solution $\widehat{W} = \widehat{w} \bf{1}$, where $\widehat{w}$ is the solution with positive imaginary part  of quadratic equation
			\begin{equation}\label{eq:w_magnetic_field}
			(M \beta)^2 + \frac{i}{2} M \beta \lambda^{-2} \cot(\beta) = \frac{i}{2} \lambda^{-2} w - w^2.
			\end{equation}
			Indeed, using $S=\widehat S\otimes R(\beta)$,
			setting $W=w\bf{1}$ in eq.~\eqref{hatWeqprimaprima}, and then simplifying by $R(-\beta)$ leads to a dimension $1$ equation of the same form as~\eqref{hatWeqprimaprima} that is equivalent to eq.~\eqref{eq:w_magnetic_field}. Then, Lemma~\ref{lem:stable_W} ensures an appropriate solution $\widehat w$ exists.
		
			Finally, since $\widehat{W}$ is proportional to the identity and $S=\widehat{S}\otimes R(\beta)$, it follows from the definition of $M$ that $M=\widehat{M}\otimes R(\beta)$ with $\widehat{M}$ the $2\times2$ matrix defined by eq.~\eqref{Ma} with $\widehat W$ set to $\widehat w$, $\Sigma=\lambda^2$, and $S$ replaced by $\widehat{S}$. Then,  
			Lemma~\ref{LLL1} shows that  Assumption {\bf{AM}} holds.
\end{enumerate}

\noindent
{\bf{Acknowledgements}.} We thank Baptiste Schubnel for very useful discussions at an early stage of the work presented in this paper.

\appendix

\section{Proof of Proposition 2.3}\label{app:CS}
In this appendix we sketch the proof of Propostion \ref{prop:characterization_coherent_POVM} of Subsect.~\ref{coh-states}, which clarifies the meaning of coherent states.

\begin{proof}
	Since a density matrix $\rho$ is a convex combination of pure states (corresponding to rank-one orthogonal projections of $\mathcal{H}$), it suffices to prove this proposition for pure states, $\rho_{\Psi} := \vert \Psi \rangle \langle \Psi \vert$, where $\Psi$ is a unit vector in $\mathcal{H}$. Then, in the state $\rho_{\Psi}$, the probability density, $f$, of $\xi$ with respect to the Lebesgue measure, $\d^{d}\xi$, is given by
	\begin{align*}
	f(\xi)=&\gamma^{-\frac32d}\det(2\Im W)^{\frac12}\times\\
	&\int_{(\mathbb R^d)^3} \e^{\i \pi^t(x-y)} \exp(\frac\i2(x-\xi)^tW(x-\xi) -\frac\i2(y-\xi)^tW^*(y-\xi))\overline{\Psi(x)}\Psi(y)\ \d^{d}\pi \d^{d} x \d^{d} y
	\end{align*}
	where $\gamma$ is a constant equal to $2\pi$. Setting $u=\frac12(x+y)$ and $v=\frac12(x-y)$, we find that
	\begin{align*}
	f(\xi)=&2^d\gamma^{-\frac32d}\det(2\Im W)^{\frac12}\times\\
	&\int_{(\mathbb R^d)^3}\e^{\i2\pi^tv}\e^{-(u-\xi)^t\Im W(u-\xi) -v^t\Im Wv +2\i v\Re W(u-\xi)}\overline{\Psi(u+v)}\Psi(u-v)\ \d^{d}\pi\d^{d} u \d^{d} v.
	\end{align*}
	Using that $2^d\gamma^{-d}\int_{(\mathbb R^d)^2}\e^{\i2\pi^tv}h(v)\d\pi \d^{d} v=h(0)$, for any integrable function $h$, we get
	\begin{align*}
	f(\xi)=&\gamma^{-\frac d2}\det(2\Im W)^{\frac12}\times\int_{\mathbb R^d}\e^{-\frac12(u-\xi)^t2\Im W(u-\xi)}|\Psi(u)|^2\d^{d} u.
	\end{align*}
	We observe that the convolution between the spectral measure of $\mathbf X$ and the appropriate Gaussian density appears in this formula. Hence the proposition is proven for $\xi= \mathbf{X}+\mathbf{Z}_{x}$.
	
	Similarly, the probability density, $g$, of $\pi$ with respect to the Lebesgue measure $\d^{d} \pi$ is given by
	\begin{align*}
	g(\pi)=&2^d\gamma^{-\frac32d}\det(2\Im W)^{\frac12}\times\\
	&\int_{(\mathbb R^d)^3}\e^{\i2\pi^tv}\e^{-(u-\xi)^t\Im W(u-\xi) -v^t\Im Wv +2\i v\Re W(u-\xi)}\overline{\Psi(u+v)}\Psi(u-v)\ \d^{d}\xi\d^{d} u \d^{d} v.
	\end{align*}
	Integration over $\xi$ yields
	\begin{align*}
	g(\pi)=&2^d\gamma^{-d}\int_{(\mathbb R^d)^2}\e^{\i2\pi^tv}\e^{-v^t(\Im W +\Re W(\Im W)^{-1}\Re W)v}\overline{\Psi(u+v)}\Psi(u-v)\ \d^{d} u \d^{d} v.
	\end{align*}
	Let $\Psi(x)=\gamma^{-\frac d2}\int_{\mathbb R^d} \e^{\i p^tx} \widehat{\Psi}(p)\d^{d} p$. Then
	\begin{align*}
	g(\pi)& = 2^{2d}\gamma^{-2d}\times\\
	&\int_{(\mathbb R^d)^4}\e^{\i 2q^tu}\e^{\i 2(\pi-p)^tv}\e^{-v^t(\Im W +\Re W(\Im W)^{-1}\Re W)v}\overline{\widehat{\Psi}(p+q)}\widehat{\Psi}(p-q)\ \d^{d} u\, \d^{d} v\, \d^{d} p\ \d^{d} q.
	\end{align*}
	As in our calculation of $f$ (integration over $\pi$ and $v$), integration over $q$ and $u$ yields
	\begin{align*}
	g(\pi)&=2^{d}\gamma^{-d}\times\\
	&\int_{(\mathbb R^d)^2}\e^{\i2(\pi-p)^tv}\e^{-v^t(\Im W +\Re W(\Im W)^{-1}\Re W)v}|\widehat{\Psi}(p)|^2\ \d^{d} v \,\d^{d} p.
	\end{align*}
	Integrating over $v$, one arrives at 

	\begin{align*}
	g(\pi)=&\gamma^{-\frac d2}\det(2\Im W +2\Re W(\Im W)^{-1}\Re W)^{\frac12}\times\\
	&\int_{\mathbb R^d}\e^{-(\pi-p)^t( {2}\Im W+
	{2}\Re W(\Im W)^{-1}\Re W)^{-1}(\pi-p)}|\widehat{\Psi}(p)|^2\d^{d} p.
	\end{align*}
	In this expression, we recognize the convolution between the spectral measure of $\mathbf P$ in the state $\Psi$ and the appropriate Gaussian density. This completes the proof of the proposition.
\end{proof}

\section{Proof of Lemma \ref{lem:rec}}\label{Pf}

	To start with, we note that Lemma~\ref{lem:stable_W} implies that $\Sigma-(2\Im \widehat W)^{-1}$ is positive-definite. Moreover, the matrix $M$ is invertible. Thus, all expressions in Lemma \ref{lem:rec} are well defined. Recalling that 
	$K=S^{-1}M^{-1}R$, which follows from the definitions of these matrices (see \eqref{K-matrix}), one sees that Eq.~\eqref{eq:rec_fwd_zeta} yields the recursion
	$$\zeta_{k+1}= S\zeta_k +M^{-1}\begin{pmatrix}
	\multirow{2}{*}{R}&0\\&0
	\end{pmatrix}\zeta_k -M^{-1}R\mathbf{q}_k\,,$$
	and applying the trivial identity ${\bf{1}}+ A^{-1}({\bf{1}}-A)=A^{-1}$ to $A=MS$, we find that
	\begin{equation}\label{eq:rec_fwd_zeta_with_M}
	\zeta_{k+1}=M^{-1}\zeta_k-M^{-1}R\mathbf q_k
	\end{equation}
	Eq. ~\eqref{eq:explicit_zeta} follows by iterating this recurrence equation.
	
	Next, we prove Eq.~\eqref{eq:rec_bwd_zeta}. The first equality follows by inverting the recurrence of Eq.~\eqref{eq:rec_fwd_zeta}. The second equality follows by identifying the right side of \eqref{eq:rec_bwd_zeta} with $M\zeta_{k+1} +R\mathbf{q}_{k}$, which is equal to $\zeta_{k}$, as is seen by inverting the recurrence in ~\eqref{eq:rec_fwd_zeta_with_M}.
	
	Finally we turn to the proof of ~\eqref{eq:rec_bwd_coherent}. By Lemmas~\ref{SLemma} and~\ref{qLemma},
	$$V_{\mathbf q_k}^*U_S^*\ket{\widehat W,\zeta_{k+1}}=V_{\mathbf q_k}U_S^{-1}\ket{\widehat W,\zeta_{k+1}}=\sqrt{\mathcal N(\mathbf{q}_k - S^{inv}_q\zeta_{k+1},C)}\ket{\widehat W,\zeta'}\,,$$
	where $S_q^{inv}:=\begin{pmatrix}{\bf{1}} &0\end{pmatrix}S^{-1}$, $C:=\Sigma+(2\Im \widehat W -\Sigma^{-1})^{-1}$ and 
	$\zeta'$ is set to $\zeta'=\zeta_k$, with $\zeta_k$ as in Eq.~\eqref{eq:rec_bwd_zeta}. 
	Indeed, Lemmas~\ref{SLemma} and~\ref{qLemma} show that the image of a coherent state $\vert \widehat{W}, \zeta_{k+1} \rangle$ under the action of 
	$U_S^* \cdot V_{\mathbf q_k}$ is a coherent state with the \textit{same} squeezing matrix $\widehat{W}$ centered near a certain phase space point $\zeta'$. Then Eq. \eqref{eq:rec_bwd_zeta} can be used to show that $\zeta'=\zeta_k$.\\
From~\eqref{eq:rec_bwd_zeta}, we deduce that
	$$K(\mathbf{q}_{k}-\xi_{k})=R(\mathbf{q}_{k} -S_q^{inv}\zeta_{k+1}).$$
	Hence,
	$$\mathbf{q}_{k} -S_q^{inv}\zeta_{k+1}=D(\mathbf{q}_{k}-\xi_{k})\,,$$
	with $D=\Sigma(\Sigma-(2\Im \widehat W)^{-1})^{-1}$, which happens to be equal to the inverse of the upper diagonal block of $S^{-1}M^{-1}$. Thus, since $S$ is a symplectic matrix and the lower diagonal block of $S^{-1}M^{-1}$ is the identity, we conclude that $\det D=\det M^{-1}$.
	
	Using the well known properties of Gaussians, we conclude that
	$$\mathcal N(\mathbf{q}_k - S^{inv}_q\zeta_{k+1},C)=\mathcal N(D(\mathbf{q}_k - \xi_{k}),C)=\det M 
	\cdot \mathcal N(\mathbf{q}_k - \xi_{k},D^{-1}C{D^t}^{-1})\,.$$
	Since $D^{-1}=\id-(\Sigma2\Im\widehat W)^{-1}=(\Sigma2\Im\widehat W-\id)(\Sigma2\Im\widehat W)^{-1}$ and $C=\Sigma2\Im\widehat W(2\Im\widehat W -\Sigma^{-1})^{-1}$, it follows that
 $ D^{-1}C{D^t}^{-1}=\Sigma-(2\Im \widehat W)^{-1}$. 	
This completes the proof of the lemma.

\section{Norm Bounds}\label{NB}

\begin{proposition}\label{prop:NB}
Suppose that $T$ is a linear operator on $ \mathbb{C}^{d} $  and $\lambda_{max} $ is its eigenvalue with maximal modulus.  Then there is a constant $C$ (independent of $n$) such that 
\begin{align}
\|  T^n \| \leq C n^{d-1}  |\lambda_{max}  |^n. 
\end{align}
\end{proposition}

\begin{proof} We use the Jordan normal form and take $n \geq d$. 
 Let  $  J(\lambda ) $ be a Jordan block of size $m$ and $\lambda$ in all diagonal entries. Then we can write  
 $J (\lambda) =   \lambda +   N_m  $ and $(N_m)^{m} = 0$.   It follows that there is a constant $C$ such that
$$
\|  (\lambda +   N_m)^n\| =  \Big  \| \sum_{k  \leq m-1 }  \begin{pmatrix}
  n \\ k    
  \end{pmatrix}   \lambda^{n-k} N_m^{ k } \Big \|  \leq C  n^{m-1} |\lambda_{max}  |^n. 
$$ 
Using the equation above and the Jordan normal form of $T$, we obtain the desired result. 
\end{proof}

\begin{center}
-----
\end{center}

\bigskip

\noindent
$^{1}$ Department of Mathematical Physics, Applied Mathematics and Systems Research Institute (IIMAS),
National Autonomous University of Mexico (UNAM), \href{mailto:ballesteros.miguel.math@gmail.com}{ballesteros.miguel.math@gmail.com}
\\[0.3em]
$^{2}$ Institut de Mathématiques de Toulouse, UMR5219, Université de Toulouse, CNRS, UPS IMT, F-31062 Toulouse Cedex 9, France, \href{mailto:tristan.benoist@math.univ-toulouse.fr}{tristan.benoist@math.univ-toulouse.fr}
\\[0.3em]
$^{3}$ Department of Mathematics, Virginia Polytechnical Institute, \href{mailto:martin.fraas@gmail.com}{martin.fraas@gmail.com}
\\[0.3em]
$^{4}$ Department of Mathematics, UC Davis
\\[0.3em]
$^{5}$ Institute of Theoretical Physics, ETH Zurich, \href{mailto:juerg@phys.ethz.ch}{juerg@phys.ethz.ch}


\begin{thebibliography}{References}

\bibitem{Gamow} George Gamow, ``Zur Quantentheorie des Atomkerns'', Physikalische Zeitschrift {\bf{51}}, 204 (1928)

\bibitem{FiTe} Rodolfo Figari and Alessandro Teta, ``Emergence of classical trajectories in quantum systems: the cloud chamber problem in the analysis of Mott (1929)'', Arch. Hist. Exact Sci. {\bf{67}}, 215-234 (2013)

\bibitem{Fr} J\"urg Fr\"ohlich, ``A Brief Review of the ``ETH- Approach to Quantum Mechanics'''', preprint 2019,  	arXiv:1905.06603; to appear in a book edited by N. Anantharaman and M. Rassias.

\bibitem{vN} John von Neumann, ``Mathematical Foundations of Quantum Mechanics'', Princeton University Press, Princeton NJ 08540, 1955, pp. 419, 420.

\bibitem{Mott} Nevill F. Mott, ``The Wave Mechanics of $\alpha$-Ray Tracks'', Proceedings of the Royal Society A{\bf{126}}, 79–84 (1929); doi:10.1098/rspa.1929.0205

\bibitem{Darwin} Charles G. Darwin, ``A collision problem in wave mechanics'', Proc. Royal Soc. (London) A{\bf{124}}, 375-394 (1929)

\bibitem{FigTeta} Rodolfo Fiagri and Alessandro Teta, ``Quantum Dynamics of a Particle in a Tracking Chamber'', Springer Briefs in Physics, Springer-Verlag, Heidelberg, New York, Dordrecht, London, 2014.

\bibitem{MK} Hans Maassen and Burkhard K\"ummerer, ``Purification of Quantum Trajectories'', Lecture Notes-Monograph Series, vol. {\bf{48}}, pp. 252–261,  Springer-Verlag, Berlin 2006.

\bibitem{BB} Michel Bauer and Denis Bernard, ``Convergence of repeated quantum non-demolition measurements and wave-
function collapse'',  Phys. Rev. A {\bf{84}} (4), 044103 (2011)

\bibitem{BBB} Michel Bauer, Tristan Benoist and Denis Bernard, ``Repeated quantum non-demolition measurements: convergence and continuous time limit.'' Ann. H. Poincaré. {\bf 14} (4), 639--679 (2013)

\bibitem{BFFS} Miguel Ballesteros, Martin Fraas, J\"urg Fr\"ohlich, Baptiste Schubnel, ``Indirect Acquisition of Information in Quantum Mechanics'', J. Stat. Phys. {\bf{162}}, 924-958 (2016); and refs. given there.

\bibitem{BCFFS} Miguel Ballesteros, Nicholas Crawford, Martin Fraas, J\"urg Fr\"ohlich, Baptiste Schubnel, ``Non-demolition measurements of observables with general spectra'', preprint 2017, arXiv:1706.09584, in: Contemporary Mathematics, vol. {\bf{717}}, ``Mathematical Problems in Quantum Physics'', p.p. 241- 256, F. Bonetto, D. Borthwick, E. Harrell, M. Loss (eds.), AMS publ., Providence RI, 2018.

\bibitem{Bourgain} Jean Bourgain, ``A Remark on the Uncertainty Principle for Hilbertian Basis'', J. Funct. Anal. {\bf{79}}, 136-143 (1988).

\bibitem{BFF} Tristan Benoist, Martin Fraas and J\"urg Fr\"ohlich, unpublished; (see Sect. 5.5 of lecture notes of a course taught by JF at LMU-Munich, December 2019)

\bibitem{BBJ} Michel Bauer, Denis Bernard and Tony Jin. ``Monitoring continuous spectrum observables: the strong measurement limit.", SciPost Phys. {\bf{5}} (2018)

\bibitem{BK}  V.P. Belavkin and V.N. Kolokol'tsov. ``Semiclassical asymptotics of quantum stochastic equations'', Theor Math Phys {\bf{89}}, 1127-1138 (1991).

\bibitem{BDK} Angelo Bassi, Detlef D\"{u}rr and Martin Kolb. ``On the long time behavior of free stochastic Schr\"{o}dinger evolutions", Rev. Math. Phys. {\bf{22}}, 55-89 (2010).

\bibitem{DerezinskiGerard} Ian Derezniski and Christian G\'erard. ``Mathematics of quantization and quantum fields", Cambridge University Press, 2013.

\bibitem{JacodProtter} Jacod, Jean and Philip Protter.`` Probability essentials.'' Springer-Verlag Berlin Heidelberg, 2004.
\end{thebibliography}
\end{document}